\documentclass[fontsize=10pt,DIV=calc,oneside,twocolumn]{scrartcl}
\usepackage[utf8]{inputenc}

\usepackage{lmodern}
\DeclareMathAlphabet{\mathpzc}{T1}{pzc}{m}{it}
\usepackage{dsfont}
\usepackage{microtype}
\usepackage[colorlinks=true]{hyperref}
\usepackage{booktabs}
\usepackage{cuted}
\usepackage{caption}
\usepackage{xspace,xparse,etoolbox}

\usepackage{enumerate}

\usepackage{amsmath}
\usepackage{amsthm}
\usepackage{amssymb}

\usepackage{stmaryrd}

\usepackage{ifthen}
\usepackage{cleveref}
\crefname{lemma}{lemma}{lemmas}
\crefname{proposition}{proposition}{propositions}
\crefname{definition}{definition}{definitions}
\crefname{theorem}{theorem}{theorems}
\crefname{conjecture}{conjecture}{conjectures}
\crefname{corollary}{corollary}{corollaries}
\crefname{example}{example}{examples}
\crefname{section}{section}{sections}
\crefname{appendix}{appendix}{appendices}
\crefname{figure}{fig.}{figs.}
\crefname{equation}{eq.}{eqs.}
\crefname{table}{table}{tables}
\crefname{item}{property}{properties}
\crefname{remark}{remark}{remarks}

\newtheorem{theorem}{Theorem}
\newtheorem{definition}{Definition}

\newtheorem{lemma}{Lemma}

\newtheorem{remark}{Remark}

\usepackage[dvipsnames]{xcolor}
\usepackage{tikz}
\usetikzlibrary{
    calc,
    positioning,
    arrows,
    shapes.misc,
    decorations.pathreplacing,
    decorations.pathmorphing,
    decorations.markings,
    decorations.text,
    patterns,
    backgrounds,
    external
}

\newif\ifopenbraceleft
\newif\ifopenbraceright
\openbraceleftfalse
\openbracerightfalse

\pgfkeys{
  /pgf/decoration/.cd,
  open/.is choice,
  open/left/.is if = openbraceleft,
  open/right/.is if = openbraceright,
  open/both/.style = {open/left = true, open/right = true},
  open/none/.style = {open/left = false, open/right = false},
  open/.initial = none
}

\pgfdeclaredecoration{open brace}{initial}
{
  \state{initial}[width=+\pgfdecoratedremainingdistance,next state=final]
  {
    \pgfpathmoveto{\pgfpointorigin}
    \ifopenbraceright
    \pgfpathmoveto{\pgfqpoint{0pt}{.5\pgfdecorationsegmentamplitude}}
    \pgfpathlineto{\pgfqpoint{3pt}{.5\pgfdecorationsegmentamplitude}}
    \pgfpathmoveto{\pgfqpoint{6pt}{.5\pgfdecorationsegmentamplitude}}    
    \pgfpathlineto{\pgfqpoint{9pt}{.5\pgfdecorationsegmentamplitude}}
    \pgfpathmoveto{\pgfqpoint{12pt}{.5\pgfdecorationsegmentamplitude}}    
    \else
    \pgfpathcurveto
    {\pgfqpoint{.15\pgfdecorationsegmentamplitude}{.3\pgfdecorationsegmentamplitude}}
    {\pgfqpoint{.5\pgfdecorationsegmentamplitude}{.5\pgfdecorationsegmentamplitude}}
    {\pgfqpoint{\pgfdecorationsegmentamplitude}{.5\pgfdecorationsegmentamplitude}}
    \fi
    {
      \pgftransformxshift{+\pgfdecorationsegmentaspect\pgfdecoratedremainingdistance}
      \pgfpathlineto{\pgfqpoint{-\pgfdecorationsegmentamplitude}{.5\pgfdecorationsegmentamplitude}}
      \pgfpathcurveto
      {\pgfqpoint{-.5\pgfdecorationsegmentamplitude}{.5\pgfdecorationsegmentamplitude}}
      {\pgfqpoint{-.15\pgfdecorationsegmentamplitude}{.7\pgfdecorationsegmentamplitude}}
      {\pgfqpoint{0\pgfdecorationsegmentamplitude}{1\pgfdecorationsegmentamplitude}}
      \pgfpathcurveto
      {\pgfqpoint{.15\pgfdecorationsegmentamplitude}{.7\pgfdecorationsegmentamplitude}}
      {\pgfqpoint{.5\pgfdecorationsegmentamplitude}{.5\pgfdecorationsegmentamplitude}}
      {\pgfqpoint{\pgfdecorationsegmentamplitude}{.5\pgfdecorationsegmentamplitude}}
    }
    {
      \pgftransformxshift{+\pgfdecoratedremainingdistance}
      \ifopenbraceleft
        \pgfpathlineto{\pgfqpoint{-12pt}{.5\pgfdecorationsegmentamplitude}}
        \pgfpathmoveto{\pgfqpoint{-9pt}{.5\pgfdecorationsegmentamplitude}}
        \pgfpathlineto{\pgfqpoint{-6pt}{.5\pgfdecorationsegmentamplitude}}
        \pgfpathmoveto{\pgfqpoint{-3pt}{.5\pgfdecorationsegmentamplitude}}
        \pgfpathlineto{\pgfqpoint{0pt}{.5\pgfdecorationsegmentamplitude}}    
      \else
      \pgfpathlineto{\pgfqpoint{-\pgfdecorationsegmentamplitude}{.5\pgfdecorationsegmentamplitude}}
      \pgfpathcurveto
      {\pgfqpoint{-.5\pgfdecorationsegmentamplitude}{.5\pgfdecorationsegmentamplitude}}
      {\pgfqpoint{-.15\pgfdecorationsegmentamplitude}{.3\pgfdecorationsegmentamplitude}}
      {\pgfqpoint{0pt}{0pt}}
      \fi
    }
  }
  \state{final}
  {}
}

\usepackage[outline]{contour}
\contourlength{2pt}

\usepackage[backend=biber,style=numeric-comp,sorting=none,url=false]{biblatex}
\bibliography{bibliography/hamilton}
\bibliography{bibliography/books}
\bibliography{bibliography/maris}

\newcommand{\C} {\mathbb C}
\renewcommand{\H} {{\ensuremath{\mathcal H}}\xspace}
\newcommand{\B} {\mathcal B}
\newcommand{\1} {\ensuremath{\mathds 1}}

\newcommand{\op}[1] {\mathbf{#1}}
\newcommand{\set}[1] {\mathrm{#1}}
\newcommand{\field}[1] {\mathds{#1}}

\newcommand{\BQEXP}{\textnormal{\textsf{BQEXP}}\xspace}
\newcommand{\QMA}{\textnormal{\textsf{QMA}}\xspace}
\newcommand{\QMAEXP}{\textnormal{\textsf{QMA\textsubscript{EXP}}}\xspace}
\newcommand{\yes}{\textnormal{\textsf{YES}}\xspace}
\newcommand{\no}{\textnormal{\textsf{NO}}\xspace}

\input{braket}

\DeclareDocumentCommand{\kdHam}{ o o }{\textsc{$(\IfValueTF{#1}{#1}{k},\IfValueTF{#2}{#2}{d})$-TILH-3D}\xspace}

\renewcommand{\L}{\Pi}
\newcommand{\Lyes}{\Pi_\yes}
\newcommand{\Lno}{\Pi_\no}

\DeclareMathOperator{\BigO}{O}

\DeclareMathOperator{\poly}{poly}
\newcommand{\lmin}{\lambda_\mathrm{min}}
\newcommand{\ii}{\mathrm{i}}
\tikzset{%
    lattice/.style={
        baseline=0,
        x={(-0.9848cm,0.1536cm)},
        y={(0.866cm,0.42cm)},
        z={(0cm,-1cm)},
        xscale=1.25,
        yscale=1.25,
        every path/.style={text centered, draw=black, line cap=round, miter limit=4.00, line width=0.700pt},
    },
    every node/.style={
        align=center
    },
    faint/.style={draw opacity=.5,fill opacity=.5},
    vertex/.style={%
        draw=#1,fill=#1
    },
    cross/.style={
        cross out,
        draw,
        minimum size=2*(3pt-\pgflinewidth), 
        inner sep=0pt,
        outer sep=0pt
    },
    intA/.style = {
        draw = Cerulean,
        opacity = .6,
        line width=12
    },
    intB/.style = {
        draw = MidnightBlue,
        opacity = .2,
        line width=12
    }
}
\DeclareDocumentCommand{\B}{ O{} m m m  }{%
    \draw[fill=black,draw=none,opacity=1/(2*(#2)*(#3)*(#4)+1),#1] (#2,#3,#4) circle (2pt);%
}
\DeclareDocumentCommand{\inline}{ m }{%
    \tikz[lattice,baseline=-0.7ex]{#1}\xspace%
}
\DeclareDocumentCommand{\Rx}{ O{} O{0} O{0} O{0} }{%
    \draw[vertex=red,opacity=1/(2*(#2)*(#3)*(#4)+1),#1] (#2,#3,#4) node[cross, line width=1pt] {};%
}
\DeclareDocumentCommand{\Ro}{ O{0} O{0} O{0} }{%
    \draw[draw=red,fill=none,opacity=1/(2*(#1)*(#2)*(#3)+1),fill=white] (#1,#2,#3) circle (1.6pt);%
}
\DeclareDocumentCommand{\R}{ O{} m m m }{%
    \draw[fill=red,draw=none,opacity=1/(2*(#2)*(#3)*(#4)+1),#1] (#2,#3,#4) circle (1.6pt);%
}
\DeclareDocumentCommand{\Ra}{ O{} m m m }{%
    \draw[fill=green!90!black,draw=none,opacity=1/(2*(#2)*(#3)*(#4)+1),#1] (#2,#3,#4) circle (1.6pt);%
}
\DeclareDocumentCommand{\$}{ O{white} m }{%
    \rlap{\contourlength{.5pt}\contour{#1}{#2}}%
    \rlap{\contourlength{1pt}\contour{#1}{#2}}%
    \rlap{\contourlength{1.5pt}\contour{#1}{#2}}%
    \contourlength{2pt}\contour{#1}{#2}%
}
\definecolor{tileRed}{HTML}{cc0000}
\DeclareDocumentCommand{\tile}{m m m m}{%
    \tikz[lattice,baseline=7pt]{
        \draw (0, 0, 0) -- (0, 1, 0) -- (1, 1, 0) -- (1, 0, 0) -- cycle;
        \begin{scope}[every path/.style={line width=.1pt}]
        \ifstrequal{#1}{1}{
            \draw[vertex=tileRed, draw=tileRed] (1, 1, 0) -- (0, 1, 0) -- (.5, .5, 0) -- cycle;
        }{}
        \node[black] at (.55, 1, 0) {\${#1}};
        \ifstrequal{#2}{1}{
            \draw[vertex=tileRed,  draw=tileRed] (0, 0, 0) -- (0, 1, 0) -- (.5, .5, 0) -- cycle;
        }{}
        \node[black] at (0, .5, 0) {\${#2}};
        \ifstrequal{#3}{1}{
            \draw[vertex=tileRed, draw=tileRed] (0, 0, 0) -- (1, 0, 0) -- (.5, .5, 0) -- cycle;
        }{}
        \node[black] at (.45, 0, 0) {\${#3}};
        \ifstrequal{#4}{1}{
            \draw[vertex=tileRed, draw=tileRed] (1, 0, 0) -- (1, 1, 0) -- (.5, .5, 0) -- cycle;
        }{}
        \node[black] at (1, .5, 0) {\${#4}};
        \end{scope}
    }%
}

\newcommand{\SL}{\blacktriangleleft}
\newcommand{\SR}{\blacktriangleright}
\newcommand{\Sr}{\ensuremath{\SR}\xspace}
\newcommand{\Sl}{\ensuremath{\SL}\xspace}

\usepackage[colorinlistoftodos]{todonotes}
\definecolor{TODOcolorJohannes}{rgb}{1.0,.78,.26}
\definecolor{TODOcolorStephen}{rgb}{.78,1.0,.33}
\definecolor{TODOcolorall}{rgb}{1.,0.,.66}

\usepackage{authblk}

\title{\vspace{-2cm}The Complexity of Translationally-Invariant Low-Dimensional Spin Lattices in 3D}
\author[1]{Johannes Bausch\thanks{jkrb2@cam.ac.uk}}
\author[2]{Stephen Piddock\thanks{stephen.piddock@bristol.ac.uk}}
\affil[1]{DAMTP, University of Cambridge}
\affil[2]{School of Mathematics, University of Bristol}

\begin{document}
\twocolumn[
\maketitle
\begin{abstract}
	In this paper, we consider spin systems in three spatial dimensions, and prove that the local Hamiltonian problem for 3D lattices with face-centered cubic unit cells, 4-local translationally-invariant interactions between spin-$3/2$ particles and open boundary conditions is \QMAEXP-complete.
	We go beyond a mere embedding of past hard 1D history state constructions, and utilize a classical Wang tiling problem as binary counter in order to translate one cube side length into a binary description for the verifier input.
    We further make use of a recently-developed computational model especially well-suited for history state constructions, and combine it with a specific circuit encoding shown to be universal for quantum computation.
	These novel techniques allow us to significantly lower the local spin dimension, surpassing the best translationally-invariant result to date by two orders of magnitude (in the number of degrees of freedom per coupling).
    This brings our models en par with the best non-translationally-invariant construction.
\end{abstract}
\vspace{.8cm}
]

\newpage

\section*{Introduction and Motivation}
Hamiltonian operators are used ubiquitously to describe physical properties of multi-body quantum systems,
and are of paramount interest for an array of disciplines ranging from theoretical computer science, to experimental and condensed matter physics.
While computer scientists are interested in the computational power of different models (e.g.\ Hamiltonian quantum computers), for physicists it is important to calculate the structure of the low-energy spectrum of quantum systems. One of the most basic, yet fundamental such question is to estimate the ground state energy of a many-body spin system with low-range interactions, formally known as the \emph{local Hamiltonian problem}.

Kitaev's seminal paper proving quantum-\textsf{NP}-hardness of the local Hamiltonian problem for the case that each interaction couples at most five spins \cite{Kitaev2002} motivated significant progress towards understanding the computational complexity that arises in different variants of the local Hamiltonian problem \cite{Kempe2006,Oliveira2008,Aharonov2009,Bravyi2006,Schuch2011,Hallgren2013,Bausch2016,Cubitt2013,Landau2013}.
These results are especially interesting from a computational perspective, answering which families of Hamiltonians are ``complicated enough'' to perform universal quantum computation \cite{Nagaj2008,Chen2011}.
Analysing the energy levels of the resulting hard instances often required the development of novel mathematical techniques, which are of independent interest e.g.\ in the context of spectral analysis of stochastic processes, or perturbation theory.
Yet from the perspective of experimental physics and material sciences, the resulting many-body quantum systems are too contrived to be of relevance; either the local spin dimension is vast, the coupling strengths vary from site to site, or the interaction graphs are not geometrically local.

Moreover, while 1D results are interesting and in a sense the most fundamental models to study (as any 1D hardness result directly implies hardness of the corresponding higher-dimensional constructions), most condensed matter systems are in fact two- or three-dimensional, and the comparison of local dimension between the best non-translationally invariant results in 1D and 2D ---8 \cite{Hallgren2013}, and 2 \cite{Oliveira2008}, respectively---indicates that moving beyond 1D allows a significant reduction of the lattice spins' dimension. It is thus a natural question to ask whether one can go beyond a simple reduction from previously-known 1D results, by exploiting these extra dimensions in a non-trivial way (i.e.\ beyond a simple embedding), but at the same time retaining nice physical properties such as a regular lattice structure and translational symmetries. We can even go further: is there a family of Hamiltonians on a physically realistic 3D crystal lattice with a QMA-hard ground state? This question is highly relevant, since such crystal structures are found ubiquitously in nature (e.g.~face-centered cubic lattices for sodium chloride, or body-centered cubic cesium chloride crystals).

\emph{
    In this paper, we prove that the local Hamiltonian problem remains computationally hard, even for a face-centered cubic lattice of spin-3/2 particles with geometrically 4-local translationally-invariant interactions, and open boundary conditions.
}

It is clear that there is always a trade-off between local dimension and interaction range: a Hermitian operator coupling $k$ spins of dimension $d$ each has $d^{2k}$ real degrees of freedom.
In 1D and for 2-local interactions, the best-known construction to date is \cite{Hallgren2013} with 8-dimensional qudits and nearest-neighbour interactions; for each coupled pair of qudits, one Hermitian operator thus has $8^2\times 8^2=16384$ free real parameters.
Enforcing translational invariance, we can regard e.g.\ \cite{Bausch2016}---nearest-neighbour interactions between spins of dimension $\approx50$---which would give roughly $(50^2)^2\approx6\times10^6$ parameters to choose from.

The construction we propose in this paper with at most 4-local interactions between spins of dimension 4 yields $4^8$ degrees of freedom, a roughly two orders-of-magnitude improvement over a straightforward embedding of the best one-dimensional construction, and en par with the best non-translationally-invariant result. It also shows that there is only about three orders of magnitude left between this construction and spin systems that we encounter every day (e.g.\ nearest-neighbour, spin 1).

\section*{Main Result}
The family of spin systems we study are described by a Hamiltonian on a face-centered cubic (cF) lattice as shown in \cref{fig:fcc}.
More precisely, we start with a finite cubic lattice $\Lambda$, where each vertex \emph{and each face} carries a 4-dimensional spin $\mathcal H_\text{loc}=\field C^4$; the overall Hilbert space $\mathcal H$ is then the tensor product of all spins.
For a geometrically local Hamiltonian $\op h$ acting on $k$ neighbouring spins (on vertices, faces, or both), we denote with $\op h^{\vec x}$ the $k$-local operator $\op h$ when offset by a lattice vector $\vec x\in\Lambda$, and acting trivially everywhere else; in case that $\op h^{\vec x}$ protrudes out of $\Lambda$, we set $\op h^{\vec x}\equiv\op 0$.
For a finite index set $I$, we consider Hamiltonians of the form
\begin{equation}\label{eq:ham}
    \op H = \sum_{i\in I}\left(c_i\sum_{\vec x\in\Lambda}\op h_i^{\vec x}\right),
\end{equation}
where each $\op h_i^{\vec x}$ couples at most $4$ spins, either within a single unit cell, or between neighbouring unit cells.
By construction, this Hamiltonian is translationally-invariant, and features open boundary conditions since we do not place special interactions at faces, edges or corners of the lattice cuboid.

The index set $I$ does not depend on the size of the lattice, and neither do any of the $\op h_i$;
we allow the $c_i=c_i(|\Lambda|)$ to depend on the system size $|\Lambda|=W\times H\times D$, but require any $c_i/c_j\in[\Omega(1/\poly|\Lambda|),\BigO(\poly|\Lambda|)]$.
This allows us to define a variant of the local Hamiltonian problem where the input is given by a description of the local terms of a Hamiltonian as in \cref{eq:ham} (i.e.\ the matrix entries of the local terms $c_i\times\op h_i$, up to polynomial precision), as well as the three side-lengths $W, H$ and $D$ of the lattice.
Moreover, we are given two parameters $\alpha<\beta$ satisfying $\beta-\alpha=\Omega(1/\poly |\Lambda|)$, and a promise that the ground state energy of $\op H$ is either smaller than $\alpha$, or larger than $\beta$.
The local Hamiltonian problem is then precisely the question of distinguishing between these two cases, and we prove the following main theorem.
\begin{theorem}\label{th:main}
The local Hamiltonian problem is \QMAEXP-complete, even for translationally-invariant 4-local interactions on a 3D face-centered cubic spin lattice (\cref{fig:fcc}) with local dimension 4, and open boundary conditions.
\end{theorem}
\QMAEXP is similar to \QMA, the quantum analogue of \textsf{NP}, but with an exponential-time verifier
instead of polynomial-time---a necessary technicality for any translationally invariant result \cite{Gottesman2009, Bausch2016}, since an $n$-qudit instance can only encode $\poly(\log n)$ bits of information (in this case the side lengths of the lattice, which encode the input in unary).
In essence, while a \QMAEXP-hard problem can be \emph{verified} in exponential time on a quantum computer, just as in the \textsf{P} vs.\ \textsf{NP} case it is not expected to be \emph{solved} as efficiently (see \cref{sec:qmaexp} for details).

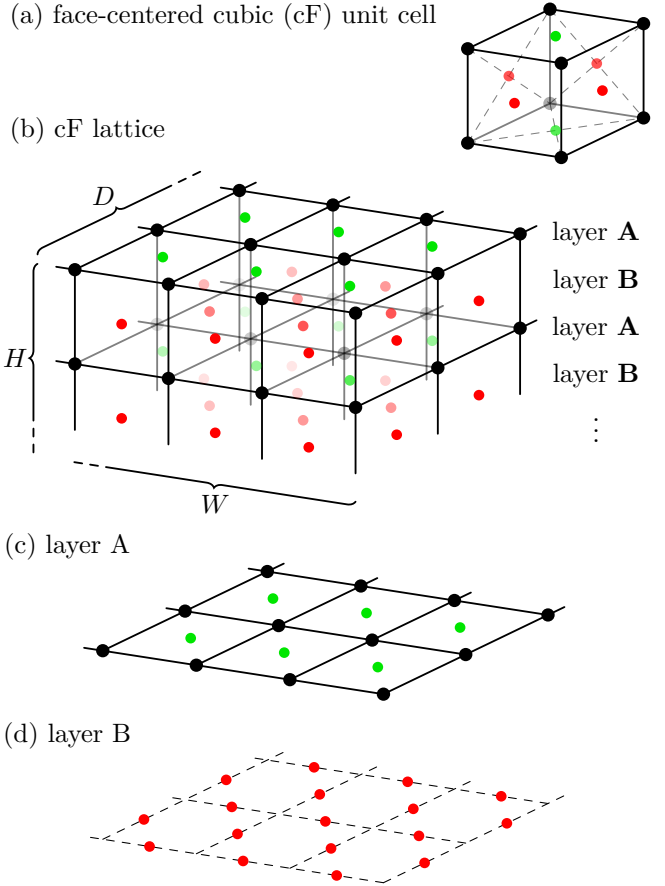
\begin{figure}
\raggedright
\begin{tikzpicture}[    
	lattice
]
    \node[right] at (3.8, 0, 0) [yshift=3.2cm] {(a) face-centered cubic (cF) unit cell};

    \coordinate (translate) at (0, 2.5, -1.6);
    \begin{scope}[shift={(translate)}]
        \coordinate (At) at (1, 1, 0);
        \coordinate (Ar) at (0, 1, 0);
        \coordinate (Ab) at (0, 0, 0);
        \coordinate (Al) at (1, 0, 0);
        \coordinate (Bt) at (1, 1, 1);
        \coordinate (Br) at (0, 1, 1);
        \coordinate (Bb) at (0, 0, 1);
        \coordinate (Bl) at (1, 0, 1);
        \draw (Ab) -- (Ar) -- (Br) -- (Bb) -- cycle;
        \draw (Ab) -- (Al) -- (Bl) -- (Bb);
        \draw (Al) -- (At) -- (Ar);
        \draw[faint] (Bl) -- (Bt) -- (Br) (Bt) -- (At);
        \draw[faint,thin,dashed]
            (Bl) -- (At) (Al) -- (Bt)
            (Bl) -- (Br) (Bb) -- (Bt)
            (Bt) -- (Ar) (At) -- (Br);
        \foreach \x/\y in {.5/0, .5/1, 0/.5, 1/.5} \R{\x}{\y}{.5};
        \foreach \x in {0, 1} \foreach \y in {0, 1} \foreach \z in {0, 1} \B \x \y \z;
        \Ra{.5}{.5}{0}; \Ra{.5}{.5}{1};
    \end{scope}
    
    \node[right] at (3.8, 0, 0) [yshift=1.7cm] {(b) cF lattice};
    
    \foreach \x/\y in {1/1,1/2,2/1,2/2,3/1,3/2}
        \draw[faint] (\x,\y,0) -- (\x,\y,1.7);
    \foreach \x/\y in {0/0,1/0,2/0,3/0,0/1,0/2}
        \draw (\x,\y,0) -- (\x,\y,1.7);
    
    \foreach \y/\z in {1/1,2/1}
        \draw[faint] (0,\y,\z) -- (3.2,\y,\z);
    \foreach \y/\z in {0/0,1/0,2/0,0/1}
        \draw (0,\y,\z) -- (3.2,\y,\z);
        
    \foreach \x/\z in {1/1,2/1,3/1}
        \draw[faint] (\x,0,\z) -- (\x,2.2,\z);
    \foreach \x/\z in {0/0,1/0,2/0,3/0,0/1}
        \draw (\x,0,\z) -- (\x,2.2,\z);
        
    \foreach \x in {0,...,3}
        \foreach \y in {0,1,2}
            \foreach \z in {0,1}
                \B \x \y \z;
    
    \foreach \x in {0,1,2}
        \foreach \y in {0,1}
            \foreach \z in {0,1} {
                \Ra{\x+.5}{\y+.5}{\z};
                \R{\x}{\y+.5}{\z+.5};
                \R{\x+.5}{\y}{\z+.5};
            }

    \node at (0, 2, 0) [xshift=1cm] {layer \textbf{A}};
    \node at (0, 2, .5) [xshift=1cm] {layer \textbf{B}};
    \node at (0, 2, 1) [xshift=1cm] {layer \textbf{A}};
    \node at (0, 2, 1.5) [xshift=1cm] {layer \textbf{B}};
    \node at (0, 2, 2) [xshift=1cm] {\vdots};
    
    \draw[decorate,decoration={open brace,amplitude=3pt,mirror,open=right}] (3.4,2,-.1) -- (3.4,0,-.1) node[midway,above,xshift=-6] {$D$};
    \draw[decorate,decoration={open brace,amplitude=3pt,mirror,open=left}] (3.4,0,0) -- (3.4,0,2) node[midway,left] {$H$};
    \draw[decorate,decoration={open brace,amplitude=3pt,mirror,open=right}] (3,0,2) -- (0,0,2) node[midway,below,yshift=-2] {$W$};
\end{tikzpicture}\hfill\\
\begin{tikzpicture}[lattice]
    \node at (-2.5, 0, 0) [yshift=2cm] {(c) layer {A}};
    \node at (-2.5, 0, 0) [yshift=-.5cm] {(d) layer {B}};

    \foreach \x in {-5,...,-2} {
        \draw (\x,1) -- (\x,3.2);
        \foreach \y in {1,2,3}
            \B \x \y 0;
    }
    \foreach \y in {1,2,3}
        \draw (-5,\y,0) -- (-1.8,\y,0);
    \foreach \x in {-5,...,-3}
        \foreach \y in {1,2}
            \Ra[opacity=1]{\x+.5}{\y+.5}{0};
            
    \foreach \x in {-5,...,-2} {
        \draw[thin,dashed] (\x,1,2) -- (\x,3.2,2);
        \foreach \y in {1,2}
            \R[opacity=1]{\x}{\y+.5}{2};
    }
    \foreach \y in {1,2,3}
        \draw[thin,dashed] (-5,\y,2) -- (-1.8,\y,2);
    \foreach \x in {-5,...,-3}
        \foreach \y in {1,2,3}
            \R[opacity=1]{\x+.5}{\y}{2};
\end{tikzpicture}
\caption{Face-centered cubic crystal lattice. All vertices and faces carry spin-3/2 particles; the red and green sublattice spins sit on the faces defined by the black lattice.}
\label{fig:fcc}
\end{figure}

We give a rigorous proof of \cref{th:main} in \cref{sec:proof}; in the following, we want to give a high-level exposition of the ideas and proof techniques which we employ. As in past hardness results, we present an explicit construction of a family of \QMAEXP-hard instances of this variant of the local Hamiltonian problem.  We will make use of two types of local terms, tiling and history state Hamiltonians, both of which have been studied extensively, but mostly independently of each other. In our work, we will utilize each method to its strength: the classical tiling terms will be used to encode the bootstrapping mechanism responsible for the large local dimension in prior work, while the history state terms will be used as a means of embedding the quantum computation part.
First we will briefly recap these methods.

\paragraph{History State Construction.}
By definition, a promise problem $\Pi$ is in \QMAEXP if there exists a \BQEXP quantum circuit---called ``verifier''---such that for any \yes-instance $l\in\Pi$, there exists a poly-sized quantum state---called ``witness''---which the verifier accepts with probability $\ge2/3$; or if $l$ is a \no-instance, all poly-sized witnesses are rejected with high probability. 
The exact constant used here is not important, as for any polynomial $p$, one can always amplify a \QMAEXP promise problem such that the distinction works with probability $\ge1-1/3^{p(|l|)}$ for an instance $l\in\Pi$ with size $|l|$ (cf.\ \cref{rem:probability-amp}).

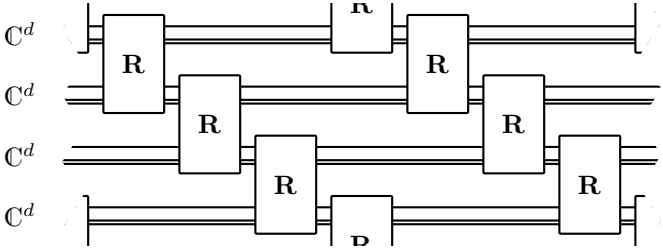
\begin{figure}[t]
  \centering
  \def\W{8}
  \def\w{1}
  \def\h{0.8}
  \def\gx{0.4}
  \def\gy{0.65}
  \begin{tikzpicture}[
      box/.style = {fill=white,draw = black,rounded corners=.5,line width=.8pt},
      snake/.style = {},
      snake2/.style = {decorate, decoration = {zigzag, amplitude = 4pt, segment length = 25pt}}
    ]

    \begin{scope}
        \clip [snake] (-0.2,0.5*\h) rectangle (\W+0.2,-3.5*\h);
        \clip [snake2] (.2,2*\h) rectangle (\W-.2,-5*\h);
    
        \newcommand{\gate}[2]{
          \path[box] (#1*\w-\gx,-#2*\h-0.5*\h-\gy) rectangle (#1*\w+\gx,-#2*\h-0.5*\h+\gy);
          \node at (#1*\w,-#2*\h-0.5*\h) {$\op R$};
        }
    
        \foreach \i in {0,...,4} {
          \draw[line width=.8pt,double] (0,-\i*\h-.1) -- (\W,-\i*\h-.1);
          \draw[line width=.8pt] (0,-\i*\h+.1) -- (\W,-\i*\h+.1);
        }
    
        \gate{0}{3}
        \gate{0}{-1}
        \gate{1}{0}
        \gate{2}{1}
        \gate{3}{2}
        \gate{4}{3}
        \gate{4}{-1}
        \gate{5}{0}
        \gate{6}{1}
        \gate{7}{2}
        \gate{8}{3}
        \gate{8}{-1}
    \end{scope}
    \foreach \i in {0,...,3} {
        \node at (-.5,-\i*\h) {$\field C^d$};
    }
  \end{tikzpicture}
  \caption{Circuit diagram of a QRM with a ring of four dimension $d$ qudits. The intuition behind proving universality for QRMs is to encode a classical (reversible) Turing machine's action into the unitary $\op R$; depending on the internal state---which is stored on classical lanes of the circuit (double lines)---a controlled unitary is applied to the pair of qubits stored on the quantum lanes (single lines).  Special flags also stored on the classical lanes indicate where the unitary $\op R$ acts non-trivially in its next round. In this way, the Turing machine can ``write out'' and apply a uniform family of quantum circuits in one go. QRMs are thus quantum Turing complete for a uniform complexity class, given that the ring scales sufficiently quickly with the input.}
  \label{fig:qrm-circuit}
\end{figure}

We further know that for any \QMAEXP promise problem, we can alternatively obtain a so-called \emph{quantum ring machine} (QRM) as verifier (\cref{lem:qmaexp-qrm}). In brief, a QRM is a fixed unitary $\op R$ on $(\field C^d)^{\otimes 2}$, which acts cyclicly on a ring of $n$ dimension $d$ qudits. Borrowing terminology from Turing machines (TM)---which are used to prove universality of the QRM model---we call the unitary $\op R$ the \emph{head} of the QRM, and the qudit ring is essentially a TM tape with cyclic boundary conditions. \Cref{fig:qrm-circuit} depicts such a QRM and its action in circuit notation.

We take a specific 2-qubit quantum gate $\op G$ and prove it to be universal, even when only applied to adjacent qubits (\cref{lem:single-gate-universality}).
Together with its inverse $\op G^\dagger$, we can thus use Solovay-Kitaev to approximate the QRM head unitary $\op R$ to within precision $\epsilon$.
Since we require that the QRM first writes out an instance $l\in\Pi$ on the ring, the resulting circuit $C_{\op R}$ has size $|C_{\op R}|=\BigO(\poly|l|\log^4(1/\epsilon))$, cf.\ \cref{lem:which-cube}. To match the QRM evolution, we repeatedly apply $C_{\op R}$ in a cyclic fashion, as described in \cref{fig:qrm-circuit}.

Keeping with tradition, we encode the circuit $C_{\op R}$ as a so-called \emph{history state} Hamiltonian. In its simplest form, such a Hamiltonian encodes transitions for each gate $\op U_i$ present in $C_{\op R}=\op U_T\cdots\op U_1$. More specifically, on the Hilbert space $\field C^{T+1}\otimes(\field C^d)^{\otimes n}$ and a basis $\{\ket j\}_j$ on $\field C^{T+1}$, we define
\newcommand\prop{\mathrm{prop}}
\begin{equation}\label{eq:histHam}
    \op H_\prop:=\sum_{t=0}^{T-1}\sum_{j}(\ket t\otimes\ket j-\ket{t+1}\otimes\op U_t\ket j)(\ \text{h.c.}\ ),
\end{equation}
where the first component of the Hilbert space stores a clock index taking track of the current step within the computation. One can verify that $\ker(\op H_\prop)$ is spanned by states of the form
\[
    \ket\Psi=\sum_{t=0}^T\ket t\otimes\ket{\psi_t}:=\sum_{t=0}^T\ket t\otimes\op U_t\cdots\op U_1\ket\phi
\]
for any initial states $\ket\phi\in(\field C^d)^{\otimes n}$. We say that the ground state is spanned by states encoding the ``history'' of the computation, meaning that the state of the computation after $t$ steps---$\ket{\psi_t}$---is entangled with the ``time'' register $\ket t$ (we want to point out, however, that this is a static problem, and the analogy with time steps is purely educational).

A large part of the overhead in terms of local dimension or interaction range present in prior constructions is due to the fact that the terms in \cref{eq:histHam} are not necessarily local. A common approach to construct a local clock is to subdivide each computational step from $\ket{\psi_t}\longmapsto\op U_t\ket{\psi_t}=\ket{\psi_{t+1}}$ into multiple intermediate steps
$$
\ket{\psi_t}\longmapsto
\ket{\psi_{t,1}}\longmapsto\ldots\longmapsto
\ket{\psi_{t,s_t}}\longmapsto
\ket{\psi_{t+1}},
$$
where no quantum gate is applied, but where some internal reordering takes place which allows each transition to act on neighbouring spins only.
This allows each gate operation in \cref{eq:histHam} to be written as a local interaction.
However, the problem remains that for each local transition rule, in order to know which gate to apply next, one has to be able to identify the current computational step locally and unambiguously (for an extensive discussion cf.\ \cite[introd.]{Bausch2016}).
Knowing when to apply which transition rule thus requires a potentially large local Hilbert space dimension, or long-range interactions.

Quantum ring machines circumnavigate part of this problem, as only a potentially much smaller circuit $C_{\op R}$ has to be applied in a periodic fashion.
However, we still need to locally store the current step \emph{within} the circuit $C_{\op R}$.
In the next section, we explain how we use diagonal Hamiltonian terms to constrain the ground space of our Hamiltonian such that the circuit description for $C_{\op R}$ is exposed at the front edge of the cuboid, in a periodically repeating fashion (cf.\ \cref{fig:cube-structure}).
More precisely, we define a diagonal Hamiltonian $\op H_\text{cl}$ with spectral gap 1, and a degenerate ground space for which any ground state of $\op H_\text{cl}+\op H_\prop$ will then be in a product configuration $\ket{\Phi_\text{cl}}\otimes\ket{\Psi}$. Here $\ket{\Phi_\text{cl}}$ is a classical product state that takes a configuration as in \cref{fig:cube-structure}: in particular a string describing $C_{\op R}$ is expressed, periodically, on the front edge.

Local terms as in \cref{eq:histHam} can then be used to access this circuit description \emph{without} any explicit knowledge of the current position within the circuit, which is implicitly given by the location on the cube where the transition rule is applied.

\paragraph{Tiling Construction.}
A tiling Hamiltonian is a local Hamiltonian on a lattice, where each term is a projector onto the complement of the allowed tiles at a specific lattice location (cf.\ e.g.\ \cite{Bausch2015}). As a simple example, consider just the 2D layer B-type sublattice from \cref{fig:fcc}, and assume that every spin is a qubit. We denote with white the state $\ket0$, and with red shading the state $\ket1$. Assume the only tiles we want to allow are the four shown in \cref{fig:cube-structure} (without rotated variants).

By writing a local term for each tile (where we order the corresponding Hilbert space $(\field C^2)^{\otimes 4}$ as a tensor product of the spin on the back, right, front, and left, respectively), we can write a diagonal projector $\op h=\1-\sum_{i=1}^4\op h_i$ on $(\field C^2)^{\otimes 4}$ such that the ground space is spanned by quantum states corresponding to the valid tiles; as an example, for the fourth tile, we write
\begin{align*}
    \op h_4&:=\ketbra1 \otimes \ketbra1 \otimes \ketbra0 \otimes \ketbra1.
\end{align*}
We can thus easily define a local Hamiltonian on the layer B-type sublattice which in its zero energy ground state encodes \emph{valid} tiling patterns, where adjacent edges match, if possible; if not, the ground state energy of the Hamiltonian will be at least 1.  More specifically, if $P$ indexes all squares with four adjacent spins, then we can write the Hamiltonian as
\begin{equation}\label{eq:static-constraint}
    \op H_\text{tiling}:=\sum_{\vec p\in P}\left(\1^{\vec p} - \sum_{i=1}^4\op h_i^{\vec p}\right)\otimes\1
\end{equation}
where $\op h_i^{\vec p}\otimes\1$ acts non-trivially only on the spins sitting on the edges of square $\vec p$.
For the aforementioned tiles, the resulting pattern is a binary counter, which can be used to translate the depth of a lattice, $D$, into a binary string representation of $D$ at the front edge (cf.\ top face in \cref{fig:cube-structure}).

The same method can equivalently be used to enforce a more complicated tiling pattern in three dimensions, especially when mixing penalty terms with different weights; for an extensive proof that the corresponding Hamiltonian ground space is indeed spanned by the best possible tiling we refer the reader to \cite[appdx.]{Bausch2015}.

\paragraph{Hard Instances for the Local Hamiltonian Problem.}
\begin{figure}[ht]
\begin{tikzpicture}
       \node[anchor=south west,inner sep=0] at (0,0) {\includegraphics[width=\columnwidth]{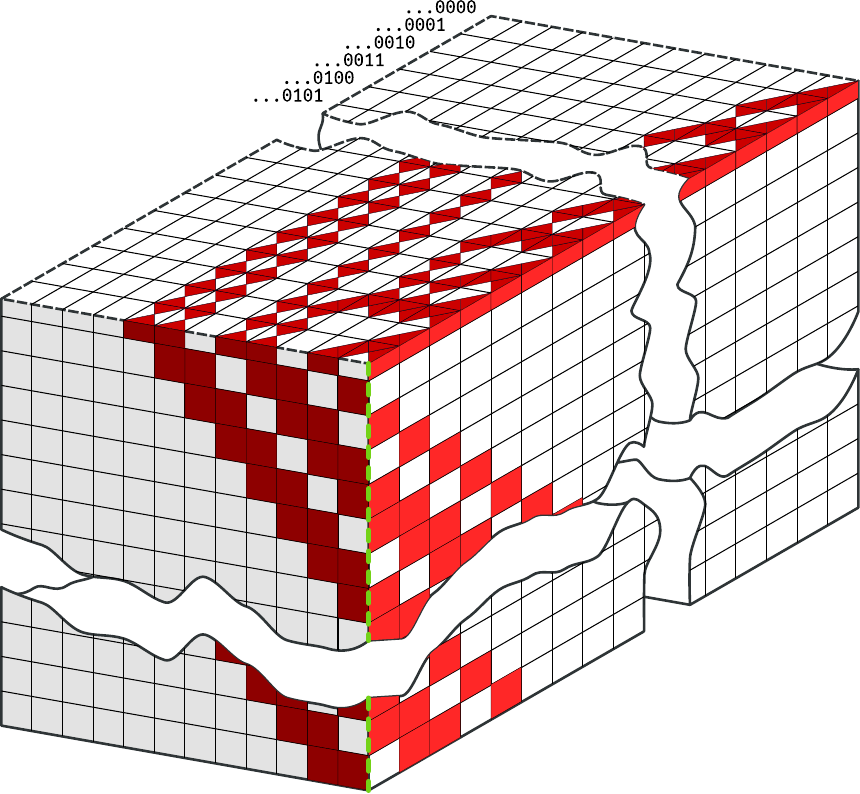}};
       \node[anchor=west] at (5,.11) {\${execution terminates}};
       \node[anchor=east,rotate=90] at (4.2,4) {\${computation edge}};
       \node[anchor=west] at (.1,4.8) {\shortstack[l]{\${binary program}\\\${description}}};
       \node[anchor=west] at (.1,8) {\shortstack[l]{\${valid tiling of plane}\\\${is a binary counter}\\\${from 0 to $D$}}};
       \node[anchor=east] at (7.2,8) {$W$};
       \node[anchor=east] at (8.7,5.2) {\${$H$}};
       \node[anchor=east] at (7.4,1.5) {$D$};
       \draw[<-,black] (4.2, .1) -- (5, .1);
\end{tikzpicture}
\caption{Structure of the ground state imposed by classical bonus and penalty terms.
Shown here is the lattice as in \cref{fig:fcc}; a cut through the top layer of the cuboid shows the layer B red sublattice depicted in fig.\ \hyperref[fig:fcc]{1 (d)}.
Coloured triangles on the top layer denote a spin on the tile edge in configuration $\ket 1$, a white tile edge stands for configuration $\ket0$; the tiles used on the top layer are the following four:
{\upshape
$$\tile0000\!\!\!
\tile0110\!\!\!
\tile1010\!\!\!
\tile1101$$}
The same colour coding is used for the squares around the sides, which label the red cF spins on the sides of the unit cells.
The dashed green front edge denotes the computation edge, where gates will be applied in the history state construction.
Observe how the same binary pattern is repeated periodically along the computation edge.}
\label{fig:cube-structure}
\end{figure}
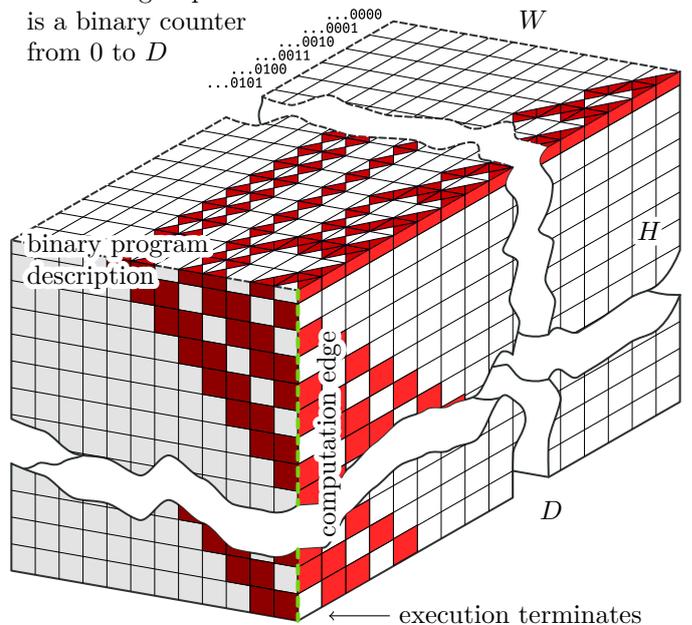

We will now explain how these two techniques---tiling and history state Hamiltonians---can be combined in order to prove \cref{th:main}.
As a first step, we define a tiling pattern to constrain all red layer B spins of the cube---apart from the top and side layers, but including the bottom layer---to a specific symbol which is used nowhere else, and which we denote with \inline\Rx.
All the following terms can then be conditioned on these red spins being either in state \inline\Rx, or not;
this allows us to distinguish between the different faces of the cuboid in a translationally-invariant way and with open boundary conditions.
This technique is commonly used in 1D (e.g.\ \cite{Gottesman2009}), and we extend it to three dimensions.

As explained in the last section, we then define four tiles which self-assemble to a binary counter; this allows us to translates the depth of the cube $D$ to a string representation of $D$ on the top front edge.
Using similar tiles on the sides of the lattice, we wind this binary string down and around the cube in an anti-clockwise direction; like that, the string---which is the binary program description of the QRM circuit $C_{\op R}$---is expressed periodically on the front edge of the lattice, which we label the \emph{computation edge}, cf.\ \cref{fig:cube-structure}.

\newcommand\Hstat{\op H_\text{stat}}
We further restrict the spins in the green layer A sublattice adjacent to this computation edge to be in a state corresponding to successive pairs of program bits. For example, if the binary program description is $p_1,p_2,p_3,p_4$, the green spins depend on $(0,p_1), (p_1,p_2), (p_2,p_3), (p_3,p_4)$ and $(p_4,0)$ respectively.
A special encoding (cf.\ \cref{tab:encoding}) allows us to translate any such binary pair into an operation to perform on the computation edge.
All constraints up to this point are diagonal in the computational basis and at most 4-local; we collect all these \emph{static} terms on the cF lattice in the Hamiltonian $\Hstat$.

In order to execute the circuit encoded by the binary string, we will assume that we are working in the ground space of $\Hstat$; any other states necessarily have energy $\ge1$.
On the black layer A spins, we partition the Hilbert space $\mathcal H_\text{loc}$ into $\field C^2\oplus\field C^2$; each spin either stores a qubit $\ket q$, or it indicates one of two ``mover'' symbols \Sl and \Sr.

We write transition rules for the two arrows, which move them around the cube according to their direction, while staying on the same layer (cf.\ \cref{sec:moving-constraints}).
Any qubit in their path is pushed down to the next layer and cycled one to the right if passed by \Sr, or one to the left if passed by \Sl.
Once an arrow arrives at the computation edge, a transition rule conditioned on the program bit pairs $(p_i,p_{i+1})$ (accessible through the green spins) performs the corresponding computational step on the two adjacent qubits.
The arrow is then re-set to the next lower level, and the whole procedure repeats.
Once the arrow returns to the computational edge and is at the bottom-most layer, there is no further forward transition; the program terminates.

Symbolically, the operations we can perform with this basic set of instructions are the following ones.
We have a quantum state of $N$ qubits $\ket{q_1}\ket{q_2}\cdots\ket{q_N}$.
In one step, we can either\ldots
\begin{enumerate}
	\item cycle the qubits clockwise, to $\ket{q_N}\ket{q_1}\cdots\ket{q_{N-1}}$,
	\item cycle them anti-clockwise, to $\ket{q_2}\cdots\ket{q_N}\ket{q_1}$,
	\item perform a universal two-qubit quantum gate $\op G$ on the first two qubits,
	\item or perform the inverse of this gate, i.e.\ $\op G^\dagger$.
\end{enumerate}
We prove in \cref{lem:single-gate-universality} that there exists such a gate $\op G$ which is universal for quantum computation, even if only applied to adjacent qubits.
Analogously to before, we collect all history state terms in the Hamiltonian $\op H_\text{hist}$.

For us, the lattice instances of interest are the ones where the binary string corresponds to a circuit approximating the head of a \QMAEXP verifier QRM, i.e.\ $C_{\op R}$ (the fact that most program strings do not represent such a QRM is not important).
In \cref{lem:which-cube} we perform a careful analysis of the approximation errors, and show that one can indeed choose height, width and depth of the lattice (depth corresponding to the encoded program, width to the ring size, and height to the run time of the verifier) such that the history state corresponds to a witness verification for any instance $l\in\Pi$, where $\Pi$ can be any promise problem in \QMAEXP.

What remains to be done is to penalize invalid history state configurations, such as multiple active symbols, or no active symbol; collect those terms in an operator $\op P$, which is the only one which will make use of the scaling freedom given in \cref{eq:histHam}.
Finally, an input penalty $\Pi_\text{in}$ for the computation ensures that some ancillas are correctly initialized for the computation, and the output penalty $\Pi_\text{out}$ raises the lowest energy for \no-instances.

Since our history state has branches (since not all transition rules we write down are completely unambiguous), we have to show that $\op H_\prop$ defines a so-called unitary labeled graph Laplacian and invoke a recently-proven variant of Kitaev's geometrical lemma for this case, \cref{lem:kitaev-ulgs}.
With a rigorous proof in \cref{sec:proof}, we can thus show that the overall 4-local translationally-invariant Hamiltonian
$$
\op H:=\Hstat + \op H_\prop+\op P+\Pi_\text{in} + \Pi_\text{out}
$$
defined on the spin-3/2 cF lattice satisfies the promise gap $\lmin(\op H)\le-\Omega(1/\poly |\Lambda|)$ if $l$ is a \yes-instance, and $\lmin(\op H)\ge0$ otherwise.
This finishes the construction, and the claim of \cref{th:main} follows.

\section*{Conclusion}
The quest for ever-more physically realistic families of \QMA hard local Hamiltonians has arguably led us to increasingly contrived constructions.
The increase in complexity necessary when going from non-translationally-invariant constructions to translational invariance is striking \cite{Gottesman2009}, and the same holds true for the effort to bring the local dimension back within reasonable range \cite{Bausch2016}.
On the other hand, almost always some fundamental new piece of machinery had to be developed, advancing our knowledge about circuit Hamiltonians: such as allowing branching to happen in the computational path, or using easier-to-implement computational models (Quantum Ring Machines), of independent interest e.g.\ in the context of adiabatic quantum computation (\cite{Wei2015}).

In our case, we combine our construction with Wang tiles, which to our knowledge have not ever been used for this purpose.
This ``outsourcing'' of part of the computation to a classical constraint satisfaction problem saves a significant amount of overhead for the control machinery surrounding the actual quantum verification procedure.
Furthermore, the single universal quantum gate could be of independent interest in other applications, as it is reasonable to imagine a physical set-up where gates can only be applied to adjacent qubits in a circuit.

In fact, our 3D construction showcases that the embedded computation need not be highly obscure, and can, in contrast, even be quite elegant, as is evident by the much lower required local dimension and the therefore much smaller number of possible interactions necessary.
By moving beyond simple spatial lattices, we can show that such structures \emph{support} the emergence of more complex behaviour, despite the intrinsic symmetry of the crystal lattices we employ. By making use of these novel features, we are able to reduce the local dimension by two orders of magnitude as compared to the best result known to date.

We suggest three concrete open problems.
\begin{enumerate}
\item
While our cube crystal structure is three-dimensional, we do not exploit its bulk structure beyond making use of its different sides. But there are small universal machines in higher dimensions (e.g.~2D or 3D Turing machines, Turmites, or cellular automata) which might be of use for improving this result further. This also leaves open the question of the required local dimension necessary for any 2D construction.
\item
The history state construction we employ still relies on a single moving ``head'' state. More recent results (cf.~\cite{Breuckmann2013}) utilize a propagating wave-front-like clock construction in 2D. A more general open question is of course whether there is any other construction, different from Feynman's circuit-to-Hamiltonian one, which would allow one to prove a QMA-hardness result for the local Hamiltonian problem (as studied e.g.\ in \cite{Bausch2016a}).
Including classical computation parts with Wang tiles is one step, but are there other, fundamentally different sets of local interactions even suitable to encode parts of a quantum computation?
\item
A bottom-up approach proving a \emph{lower} bound on the local dimension (or locality) of the interactions would be an alternative route to new insights into the \textsc{local Hamiltonian} problem.
We want to emphasize that there is not much space left for any optimization: as mentioned in the introduction, our construction allows each coupling to have $\approx10^4$ free parameters; by the same benchmark, physically realistic spin lattices found in nature allow somewhere around $(3\times 3)^2\approx 80$ different couplings.

Recent results show that e.g.\ 1D gapped Hamiltonian ground states can be approximated efficiently (i.e.\ in randomized poly-time, cf.~\cite{Landau2013}), but since history state constructions have a spectral gap that closes inverse-polynomially with the runtime of the encoded computation, a lower bound on the required local dimension remains open.
\end{enumerate}

\paragraph{Acknowledgements.}
We are very grateful for discussions with M\=aris Ozols, who contributed \cref{lem:single-gate-universality}.
J.\,B. acknowledges support from the German National Academic Foundation and the EPSRC
(grant no. 1600123).
S.\,P. was supported by the EPSRC.


\renewcommand*{\bibfont}{\small}
\printbibliography

\clearpage
\appendix

\onecolumn
\section{Quantum Complexity Classes}\label{sec:qmaexp}
In order to rigorously define the complexity classes \BQEXP and \QMAEXP, we need to understand the notion of a \emph{uniform circuit family}. Following and referring the reader to \cite{Bernstein1997,Watrous2012} for terminology,
we give the following definition.
\begin{definition}[Uniform family of quantum circuits]
	Let $f:\field N\rightarrow\field N$ be a function. A family of quantum circuits $(C_n)_{n\in\field N}$ is called \emph{$f$-uniform} if
	\begin{enumerate}[i.]
		\item each $C_n$ acts on $n$ qubits and has a distinct output qubit,
		\item each $C_n$ requires at most $f(n)$ additional ancillas $\ket 0$,
		\item each $C_n$ is composed of at most $f(n)$ gates from some universal gate set and
		\item there exists a classical Turing Machine which, on input $1^n$ produces a description of $C_n$ in fewer than $f(n)$ steps.
	\end{enumerate}
\end{definition}

A \emph{promise problem} is a pair of disjoint sets $(\Lyes,\Lno)$, corresponding to input strings for \yes and \no instances of a set of problem instances $\L=\Lyes\cup\Lno\subseteq\{0,1\}^*$. We are interested in the quantum generalization of \textsf{EXPTIME} and \textsf{NEXP}---\textsf{P} and \textsf{NP} with exponential runtime.
\begin{definition}[\BQEXP]
	A promise problem $(\Lyes,\Lno)$ is in the complexity class \BQEXP, bounded-error quantum $\exp$-time, if there exists an $\exp$-uniform family of quantum circuits $(C_n)_{n\in\field N}$ such that
	\[
		\mathrm{Pr}(C_n(s)=\yes)\ge\frac23\quad\text{for }s\in\Lyes
		\quad\text{and}\quad
		\mathrm{Pr}(C_n(s)=\yes)\le\frac13\quad\text{for }s\in\Lno,
	\]
	where $C_n(s)$ denotes the distribution obtained from executing $C_n$ on input $s\in\L$ of size $n=|s|$ and measuring the output qubit.
\end{definition}

\begin{remark}\label{rem:probability-amp}
It is a well-known fact---see \cite[prop.~3]{Watrous2012}---that one can amplify the accept and reject probability of $2/3$ to any $1-2^{-p(n)}$ for any fixed polynomial $p$.
\end{remark}

\QMAEXP is then the class of promise problems, for which any \yes or \no answer can be verified with a \BQEXP verifier.
\begin{definition}[\QMAEXP]\label{def:qmaexp}
	A promise problem $(\Lyes,\Lno)$ is in \QMAEXP, quantum Merlin-Arthur, if there exists an $\exp$-uniform family of quantum circuits, called verifier, each of which with an $\exp$-sized witness as input, such that for $l\in\Lyes$ there exists a witness such that the verifier accepts with probability at least $2/3$, and for $l\in\Lno$, the verifier accepts with probability at most $1/3$, for any witness.
\end{definition}
The last two conditions on accepting \yes and rejecting \no instances is also known as \emph{completeness} and \emph{soundness}, respectively.
Probability amplification can be directly translated from their \BQEXP counterparts, cf.\ \cref{rem:probability-amp}.

\section{The Local Hamiltonian Problem}
We regard Hermitian operators acting on a multipartite Hilbert space $\H=(\C^d)^{\otimes n}$, i.e.\ on $n$ qudits, each of \emph{local dimension} $d$. We label subsystems of \H by an ordered tuple $\set A\subseteq\{1,\ldots,n\}$. For a $k$-qudit Hamiltonian $\op h$ for some $k\le n$ and some subset $\set A$, we denote with $\op h_\set A$ the operator that acts as $\op h$ on all qudits labelled by $\set A$, and as identity---\1---everywhere else.
\begin{definition}
	A Hermitian operator $\op H$ on Hilbert space $\H=(\C^2)^{\otimes n}$ is called \emph{$k$-local} if $\op H=\sum_i\op h_i$, where we require that there exists a family of subsystems $(\set A_i)_i$ of \H such that $|\set A_i|\le k\ \forall k$, and $\op h_i=(\op h'_i)_\set{A_i}\ \forall i$.
\end{definition}

If the Hilbert space \H is translationally-invariant---e.g. a lattice $\H^{\otimes\Lambda}$---then we say that a Hamiltonian on this space exhibits translational invariance if it follows the same symmetry, i.e. that the interactions between equivalent lattice sides are identical. This allows us to define the following variant of the local Hamiltonian problem, where we follow the naming convention in \cite{Gottesman2009}, i.e. we abbreviate translationally invariant local Hamiltonian as TILH.

\begin{definition}[\kdHam]\label{def:khamiltonian}\leavevmode
	Let $\Lambda(L,M,N)$ be a 3D lattice of side lengths $L,M$ and $N$, all $\le n$, with not necessarily trivial unit cell (e.g.\ cF, cI). Let $\op H = \sum_{i,j,k} \op h_{i,j,k}$ be a 3D translationally-invariant and geometrically $k$-local Hamiltonian on the lattice qudits $(\field C^d)^\Lambda$.
	\begin{description}
		\item[Input.] Specification of the lattice size $L,M,N$, and the matrix entries of $\op h$, up to $\BigO(\poly n)$ bits of precision.
		\item[Promise.] The operator norm of each local term is bounded, $\|\op h\|\le\poly n$ and either $\lmin(\op H)\le\alpha$ or $\lmin(\op H)\ge\beta$, where $\lmin(\op H)$ denotes the smallest eigenvalue of $\op H$ and $\beta - \alpha \leq 1/p(n)$ for some polynomial $p(n)$.
		\item[Output.] \yes if $\lmin(\op H)\le\alpha$, otherwise \no.
	\end{description}
\end{definition}

\section{Ring Machines}
Instead of working with quantum circuits or Turing Machines directly, we will work with a computational model known as \emph{Quantum Ring Machine}.
\begin{definition}[Quantum Ring Machine]
	A \emph{Quantum Ring Machine}---\emph{QRM} for short---is a tuple $(\op U,n\ket{q_i},\H_f)$ of a unitary operator acting on a pair of qudits $(\field C^d)^{\otimes 2}$, and some $n\in\field N$.
    $n\in\field N$ specifies the number of qudits on the ring $\H:=(\field C^d)^{\otimes n}$.
    Starting out from the initial configuration $\ket{q_i}\in\H^{\otimes n}$, we cyclicly apply the unitary $\op U$ to two adjacent ring sites until the reduced density matrix on one qudit is completely supported in a halting subspace $\H_f\subsetneq\field C^d$; before that, the overlap of any qudit with $\H_f$ is zero.
\end{definition}
QRMs have the distinct advantage of being simple to specify locally---like circuits, but unlike Turing Machines---whilst maintaining a straightforward evolution---in contrast to circuits, which can have a very complex global structure. The similarity between QRMs and TMs is deliberate, as it allows to extend the notions of halting or runtime in a very straightforward manner. A family of QRMs $(\op U, n)_{n\in\set I}$---labelled by some index set $\set I\subset\field N$---is called $\exp$-time terminating if it halts on any possible input specified on the tape, and if the number of rounds the unitary $\op U$ makes on the tape is upper-bounded by some function $f(n)$ for all $n\in\set I$, where $f(n)=\BigO(\exp n)$. Similar to halting, we specify \emph{accepting} and \emph{rejecting} configurations as special subspaces of \H.

If we want to perform a verifier computation with the QRM, we simply leave part of the tape unconstrained as a witness, however much is required by the verification.
Central to this work is the following lemma, see \cite[cor.~34]{Bausch2016}.
\begin{lemma}\label{lem:qmaexp-qrm}
	Let $(\Lyes,\Lno)$ be a \BQEXP promise problem. Then there exists an $\exp$-time terminating family of QRMs $(M_s)_s$ such that for $l\in\Lyes$, there exists a witness $w$, such that the QRM $M_{|l|}(l,w)$ accepts with probability at least $2/3$.
    Analogously, \no instances are accepted with probability at most $1/3$.
\end{lemma}

\section{Kitaev's History State Constructions with Branching}
By embedding a \BQEXP-complete QRM into a local Hamiltonian, which allows the execution of a \QMAEXP verifier circuit, we want to construct a family of \QMAEXP-hard \kdHam instances. Using $i$ as a multiindex to sum over all 3D lattice sides of $\Lambda$, the Hermitian operators that we regard can be written in the form
\begin{equation}\label{eq:history-state-ham}
\op H=\op P+\op T=\sum_i\op p_i + \sum_i\op t_i,
\end{equation}
where the local terms in $\op P$ are diagonal projectors in the computational basis---i.e.~classical---and the terms in $\op T$ are so-called transition rules. Any transition rule always takes the form
\begin{align}\label{eq:transition-term}
	\op t
	&=\sum_{\ket e}(\ket a\otimes\ket e-\ket b\otimes\op U\ket e)(\bra a\otimes\bra e-\bra b\otimes\bra e\op U^\dagger) \\
	&\equiv (\ketbra a+\ketbra b)\otimes\1 - \ketbra{a}{b}\otimes\op U-\ketbra{b}{a}\otimes\op U^\dagger,
\end{align}
where $\ket a,\ket b$ are basis vectors in some subspace $\H_c$---which we call classical---and the $\ket e$ label a basis of some different---quantum---subspace $\H_q$. $\op U\in\set{SU}(\H_q)$ is a unitary operator on this quantum subspace. An operator $\op T$ made up of such transition rules can be thought of a generalized Laplacian for a simple graph with unitary edge labels, formally defined as follows.
\begin{definition}\label{def:ulg}
	A \emph{unitary labeled graph} (ULG) is an undirected graph $G=(\set V,\set E)$, a Hilbert space $\H_v$ for each graph vertex $v\in\set V$ and a function $g:\set E\longrightarrow\bigcup_v\mathcal B(\H_v)$, assigning a unitary operator to each edge $e\in\set E$, where $g(e)\in\mathcal B(\H_a)$ if $e=(a,b)$.
\end{definition}
In particular, the Hilbert spaces attached to two vertices are necessarily isomorphic if the vertices are connected. The associated Hamiltonian is then simply defined as a sum of transition rules of the form \cref{eq:transition-term} for each edge of the graph. We refer the reader to \cite[ch.~5]{Bausch2016} for details and a simple example.

If the product of unitaries along any loop within the graph is the identity operator---where we flip $\op U$ to $\op U^\dagger$ in case we march against an edge direction---we call the ULG \emph{simple}. A simple ULG has the following important property.
\begin{lemma}
	The associated Hamiltonian of a simple and connected ULG is unitarily equivalent to copies of the underlying graph's Laplacian $\Delta$, i.e.~there exists a unitary $\op W$ such that $\op W\op H\op W^\dagger=\Delta\otimes\1_n$, where $n=\dim\H_q$.
\end{lemma}
\begin{proof}
	Cf.~\cite[lem.~41]{Bausch2016}.
\end{proof}
Note that this extends to non-connected ULGs in a straightforward manner, as the associated Hamiltonian is block-diagonal in the connected graph components. With this, we can formulate a variant of \emph{Kitaev's geometrical lemma}, which allows us to analyse the spectra of the Hamiltonians we construct.
\begin{lemma}\label{lem:kitaev-ulgs}
	Take a history state Hamiltonian of the form \cref{eq:history-state-ham}, where $\op T$ is the associated Hamiltonian of some simple connected ULG with Hilbert space $\H_q$ for all vertices $v\in\set V$. We require that $\op P=\sum_{p\in\set P}\Pi_p\otimes\Pi_{p,q}$, where $\Pi_p$ is a projector on some vertex $p\in\set P\subseteq\set V$, and the $\Pi_{p,q}$ are projectors on subspaces of $\H_q$. Then $\lmin(\op H)=\mu\Omega(1/|\set V|^3)$, where $\mu=\min\{\lmin(\Pi_{p_i,q}\op U_{ij}\Pi_{p_j,q}):p_i,p_j\in\set P\}$ and $\op U_{ij}$ is the product of unitaries of a path connecting vertices $p_i$ and $p_j$.
\end{lemma}
\begin{proof}
	Cf.~\cite[lem.~44]{Bausch2016}.
\end{proof}

\section{Single Gate Universality}
In order to execute the QRM, we have to be able to cyclicly apply the QRM head unitary on a pair of qudits. Since we will be working with qu\emph{b}its in our construction, we embed each such qu\emph{d}it into a list of qu\emph{b}its, and approximate the 2-local qu\emph{d}it unitary using a special 2-local unitary gate $\op G$, which can act on any two neighbouring qubits.
In order to apply Solovay-Kitaev to the QRM head unitary and approximate it with a $\BigO(\log(1/\epsilon))$ gate count (as opposed to $\sim1/\epsilon$), we have to be able to apply both $\op G$ and its inverse, $\op G^\dagger$; however, in Solovay-Kitaev, the requirement is that those two gates can be applied to \emph{any pair} of qubits, whereas in our construction---as will become clear later---we can only ever apply either gate to neighbouring qubits.

It suffices to prove that $\op G$ is universal when applied to adjacent qubits, which is what the following lemma shows.

\begin{lemma}[Ozols \cite{OzolsPC}]\label{lem:single-gate-universality}
    Define the 2-qubit unitary $\op G:=\exp(\ii\op H)$ with
    \[
        \op H:=\sigma_x\otimes\1+1\otimes\sigma_z+\sigma_x\otimes\sigma_x+\sigma_z\otimes\sigma_z=
        \begin{pmatrix}
        	2 & 0  & 1 & 1 \\
        	0 & -1 & 1 & 1 \\
        	1 & 1  & 0 & 0 \\
        	1 & 1  & 0 & 0
        \end{pmatrix}.
    \]
    Then the unitaries $\{\op G_{k,k+1\mod l}:k=0,\ldots,l-1\}$ generate a dense subset of $\set{SU}(2^l)$ for all $l\ge 3$, where the subscript denotes where the unitaries act.
\end{lemma}
    \begin{table}
        \centering
        \begin{tabular}{r|rrrrrrrrrrr}
        	   &  1 &  2 &  3 &  4 &  5 &  6 &  7 &  8 &  9 & 10 & 11 \\ \midrule
        	 2 &  3 &    &    &    &    &    &    &    &    &    &  \\
        	 3 &  4 &  5 &    &    &    &    &    &    &    &    &  \\
        	 4 &  6 &  7 &  8 &    &    &    &    &    &    &    &  \\
        	 5 &    &  9 & 10 & 11 &    &    &    &    &    &    &  \\
        	 6 & 12 & 13 & 14 & 15 & 16 &    &    &    &    &    &  \\
        	 7 &    & 17 & 18 & 19 & 20 & 21 &    &    &    &    &  \\
        	 8 &    &    & 22 & 23 & 24 & 25 & 26 &    &    &    &  \\
        	 9 &    & 27 & 28 & 29 & 30 & 31 & 32 & 33 &    &    &  \\
        	10 &    &    & 34 & 35 & 36 & 37 & 38 & 39 & 40 &    &  \\
        	11 &    &    &    & 41 & 42 & 43 & 44 & 45 & 46 & 47 &  \\
        	12 &    & 48 & 49 &    & 50 &    & 51 & 52 & 53 & 54 & 55 \\
        	13 &    & 56 & 57 & 58 & 59 & 60 & 61 & 62 &    & 63 &
        \end{tabular}
        \caption{A Linearly independent set of generators for $\mathfrak{su}(8)$ in terms of nested commutators of $\op H_1$ and $\op H_2$. For example, $\op H_{42}:=\ii[\op H_{11},\op H_5].$}
        \label{tab:lie-entries}
    \end{table}
    
\begin{proof}
    Since 3-qubit unitaries generate a dense subset of $\set U(2^l)$ when applied to adjacent qubits, it suffices to prove the claim for $l=3$. The proof follows techniques in Lie algebra \cite{Childs2010}.
    Define $\op H_1:=\op H\otimes\1_2$ and $\op H_2:=\1\otimes\op H$, and let $\mathcal L(\op H_1,\op H_2)$ be the Lie algebra generated by these two elements. For $j=3,\ldots,63$, we set $\op H_j:=\ii[\op H_{r_j},\op H_{c_j}]$, where $r_j$ and $c_j$ are the row and column numbers of entry $j$ in \Cref{tab:lie-entries}. One can verify---using a computer algebra system---that the matrices $\{\op H_1,\ldots,\op H_{63}\}$ are linearly independent, and traceless by construction. Since $\dim\mathfrak{su}(8)=63$, they furthermore span the entire algebra, and the claim follows.
\end{proof}

\section{The Cube}\label{sec:cube}
\subsection{Overview}
We work with a face-centered cubic lattice of side lengths $D\times H\times W$, as shown in \cref{fig:cube-structure}. At each vertex we place a 4-dimensional spin with local Hilbert space $\mathcal H_\text{loc}=\field C^4$, and we want to define a 4-local Hamiltonian on the lattice which embeds the evolution of a \QMAEXP verifier.
Our construction comprises the following three main steps.

\paragraph{Binary Counting.}
We construct a 2D tileset which lives on the top face of the cuboid, and translates the cuboid depth $D$ into a binary description of $D$ on the top front edge, which is of size $\log_2 D$. This binary string encodes a circuit $C$ according to \cref{tab:encoding} and \cref{fig:circuit-to-desc}.

\paragraph{Shuffling the Program.}
Using another 2D tileset, we cyclicly shuffle this circuit program around the sides of the cuboid and wind it down diagonally as shown in \cref{fig:cube-structure}. The front edge---marked in red---is the computation edge and will periodically see the entire binary description of the program.

\paragraph{Performing Gates.}
On the sides of the cuboid, we superpose a layer of qubits.
Labelling the qubits around the top edge of the cube with $\ket{q_1}\ket{q_2}\cdots\ket{q_N}$, we define transition rules which allow us to perform one of the following four operations:
\begin{enumerate}
	\item cycle the qubits clockwise, to $\ket{q_N}\ket{q_1}\cdots\ket{q_{N-1}}$,
	\item cycle them anti-clockwise, to $\ket{q_2}\cdots\ket{q_N}\ket{q_1}$,
	\item perform a universal two-qubit quantum gate $\op G$ on the first two qubits,
	\item or perform the inverse of this gate, i.e.\ $\op G^\dagger$, on the first two qubits.
\end{enumerate}
 
Once any gate operation is performed, all qubits are swapped with the ones on the next-lower level while cycling them in the direction specified.
On the next layer, the same procedure repeats until the execution terminates after $H$ steps ($H$ being the height of the cube).
The history state construction for one operation above thus requires $2\times (D+W)$ steps.

 The necessary transition rules are described in detail in \cref{sec:comp}, and \cref{tab:encoding} describes how the binary program description on the computation edge is interpreted as one of the four actions above at each level. \Cref{fig:circuit-to-desc} shows how any circuit can be encoded in this way.
Observe that due to the winding program description---which is exposed periodically at the front edge---we necessarily apply the same circuit over and over again. In between each appearance of the description of $C_{\op R}$ on the computation edge, the string of zeroes does not implement any gates or move the tape in either direction.  Naturally, this is precisely the evolution of a Quantum Ring Machine.

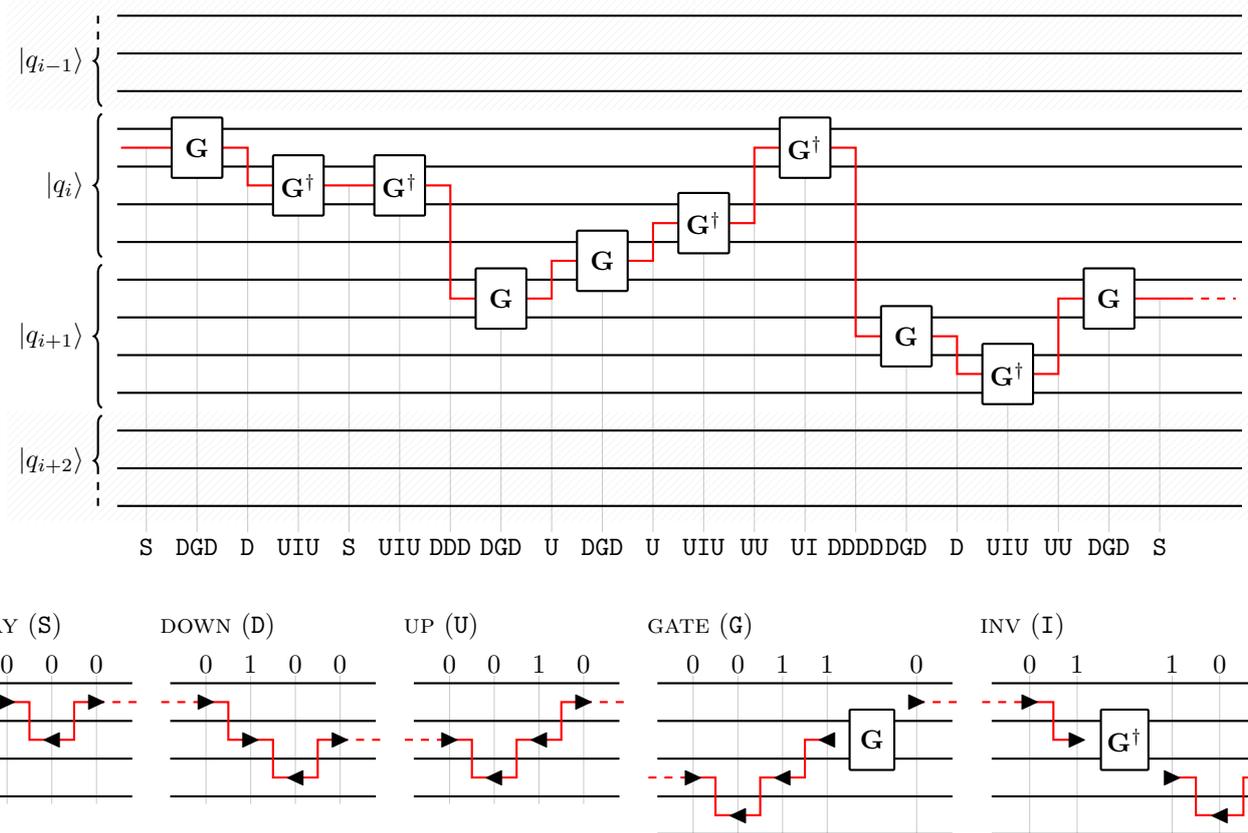
\begin{figure}
  \centering
  \begin{tikzpicture}[
      draw=black,line width=.8pt,
      box/.style = {rounded corners=.5},
      head/.style = {red},
      code/.style = {black!20}
    ]
    \newcommand{\gate}[3]{
      \draw[fill=white,box] (#1,#2-1/2-.3) rectangle (#1+1,#2-1/2+1+.3);
      \node[align=center] at (#1+.5,#2) {#3};
    }
    \fill[black,opacity=.2,pattern=north east lines] (-1.75, 4.75) rectangle (14.75, 6.2);
    \fill[black,opacity=.2,pattern=north east lines] (-1.75, -.7) rectangle (14.75, .75);

    \foreach\i in {-1,...,12}
    {\draw[] (-.3,\i/2) -- (14.5,\i/2);}

    \draw[decorate,decoration={open brace,amplitude=3pt,mirror,open=right}] (-.5, 6) -- (-.5, 4.8) node[midway,left,xshift=-3] {$\ket{q_{i-1}}$};
    \draw[decorate,decoration={brace,amplitude=3pt,mirror}] (-.5, 4.7) -- (-.5, 2.8) node[midway,left,xshift=-3] {$\ket{q_{i}}$};
    \draw[decorate,decoration={brace,amplitude=3pt,mirror}] (-.5, 2.7) -- (-.5, .8) node[midway,left,xshift=-3] {$\ket{q_{i+1}}$};
    \draw[decorate,decoration={open brace,amplitude=3pt,mirror,open=left}] (-.5, .7) -- (-.5, -.5) node[midway,left,xshift=-3] {$\ket{q_{i+2}}$};    

    \DeclareDocumentCommand{\conn}{ r() r() }{
        \draw[head] (#1) -| ($(#1) !.5! (#2)$) |- (#2);
    }
    \DeclareDocumentCommand{\desc}{ m m m O{} }{
        \draw[code, xshift=.5cm] (#1,#2) -- (#1,-2.2);
        \draw[xshift=.5cm] (#1,-2.6) node[black, align=center] {\texttt{#3}};
    }

    \begin{scope}[xshift=-.25cm,yshift=.25cm, xscale=1/1.5, yscale=1/2]
        \conn (0,8) (1,8);
        \conn (2,8) (3,7);
        \conn (4,7) (5,7);
        \conn (6,7) (7,4);
        \conn (8,4) (9,5);
        \conn (10,5) (11,6);
        \conn (12,6) (13,8);
        \conn (14,8) (15,3);
        \conn (16,3) (17,2);
        \conn (18,2) (19,4);
        \conn (20,4) (21,4);
        \draw[head,dashed] (21,4) -- (22,4);
        
        \begin{scope}[on background layer,thin]
        \desc 0 8 {S}
        \desc 1 8 {DGD}
        \desc 2 7 {D}
        \desc 3 7 {UIU}
        \desc 4 7 {S}
        \desc 5 7 {UIU}
        \desc 6 4 {DDD}
        \desc 7 4 {DGD}
        \desc 8 4 {U}
        \desc 9 5 {DGD}
        \begin{scope}[xshift=10cm] 
            \desc 0 5 {U}
            \desc 1 6 {UIU}
            \desc 2 6 {UU}
            \desc 3 8 {UI}
            \desc 4 3 {DDDD}
            \desc 5 3 {DGD}
            \desc 6 2 {D}
            \desc 7 2 {UIU}
            \desc 8 2 {UU}
            \desc 9 4 {DGD}
            \desc{10}{4}{S}
        \end{scope}
        \end{scope}

        \gate{1}{8}{$\op G$}
        \gate{3}{7}{$\op G^\dagger$}
        \gate{5}{7}{$\op G^\dagger$}
        \gate{7}{4}{$\op G$}
        \gate{9}{5}{$\op G$}
        \gate{11}{6}{$\op G^\dagger$}
        \gate{13}{8}{$\op G^\dagger$}
        \gate{15}{3}{$\op G$}
        \gate{17}{2}{$\op G^\dagger$}
        \gate{19}{4}{$\op G$}
    \end{scope}
  \end{tikzpicture}
  \\[.5cm]
  \begin{tikzpicture}[
      draw=black,line width=.8pt,
      box/.style = {rounded corners=.5},
      head/.style = {red},
      code/.style = {black!20},
      xscale=1/1.7, yscale=1/2,
      baseline=0
  ]
  \newcommand{\gate}[3]{
    \draw[fill=white,box] (#1,#2-1/2-.3) rectangle (#1+1,#2-1/2+1+.3);
    \node[align=center] at (#1+.5,#2) {#3};
  }
  \DeclareDocumentCommand{\conn}{ r() r() }{
    \draw[head] (#1) -| ($(#1) !.5! (#2)$) |- (#2);
  }  
  \node[right] at (-.25, 4) {\textsc{stay} (\texttt{S})};
  
  \foreach\x/\c in {1/0, 2/0, 3/0} {
    \draw[on background layer,thin,black!20] (\x,2.7) -- (\x,-.7);
    \node at (\x,3) {\c};
  }
  \foreach\i in {0,...,3}
     {\draw[] (.2,\i-.5) -- (3.8,\i-.5);}
  
  \draw[red,dashed] (0, 2) -- (1, 2);
  \conn (1, 2) (2, 1);
  \conn (2, 1) (3, 2);
  \draw[red,dashed] (3, 2) -- (4, 2);
  
  \node at (1, 2) {\Sr};
  \node at (2, 1) {\Sl};
  \node at (3, 2) {\Sr};
  \end{tikzpicture}
  \begin{tikzpicture}[
      draw=black,line width=.8pt,
      box/.style = {rounded corners=.5},
      head/.style = {red},
      code/.style = {black!20},
      xscale=1/1.7, yscale=1/2,
            baseline=0
  ]
  \newcommand{\gate}[3]{
    \draw[fill=white,box] (#1,#2-1/2-.3) rectangle (#1+1,#2-1/2+1+.3);
    \node[align=center] at (#1+.5,#2) {#3};
  }
  \DeclareDocumentCommand{\conn}{ r() r() }{
    \draw[head] (#1) -| ($(#1) !.5! (#2)$) |- (#2);
  }  
  \node[right] at (-.25, 4) {\textsc{down} (\texttt{D})};
  
  \foreach\x/\c in {1/0, 2/1, 3/0, 4/0} {
    \draw[on background layer,thin,black!20] (\x,2.7) -- (\x,-.7);
    \node at (\x,3) {\c};
  }
  \foreach\i in {0,...,3}
     {\draw[] (.2,\i-.5) -- (4.8,\i-.5);}
  
  \draw[red,dashed] (0, 2) -- (1, 2);
  \conn (1, 2) (2, 1);
  \conn (2, 1) (3, 0);
  \conn (3, 0) (4, 1);
  \draw[red,dashed] (4, 1) -- (5, 1);
  
  \node at (1, 2) {\Sr};
  \node at (2, 1) {\Sr};
  \node at (3, 0) {\Sl};
  \node at (4, 1) {\Sr};
  \end{tikzpicture}  
  \begin{tikzpicture}[
      draw=black,line width=.8pt,
      box/.style = {rounded corners=.5},
      head/.style = {red},
      code/.style = {black!20},
      xscale=1/1.7, yscale=1/2,
            baseline=0
  ]
  \newcommand{\gate}[3]{
    \draw[fill=white,box] (#1,#2-1/2-.3) rectangle (#1+1,#2-1/2+1+.3);
    \node[align=center] at (#1+.5,#2) {#3};
  }
  \DeclareDocumentCommand{\conn}{ r() r() }{
    \draw[head] (#1) -| ($(#1) !.5! (#2)$) |- (#2);
  }  
  \node[right] at (-.25, 4) {\textsc{up} (\texttt{U})};
  
  \foreach\x/\c in {1/0, 2/0, 3/1, 4/0} {
    \draw[on background layer,thin,black!20] (\x,2.7) -- (\x,-.7);
    \node at (\x,3) {\c};
  }
  \foreach\i in {0,...,3}
     {\draw[] (.2,\i-.5) -- (4.8,\i-.5);}
  
  \draw[red,dashed] (0, 1) -- (1, 1);
  \conn (1, 1) (2, 0);
  \conn (2, 0) (3, 1);
  \conn (3, 1) (4, 2);
  \draw[red,dashed] (4, 2) -- (5, 2);
  
  \node at (1, 1) {\Sr};
  \node at (2, 0) {\Sl};
  \node at (3, 1) {\Sl};
  \node at (4, 2) {\Sr};
  \end{tikzpicture}
  \begin{tikzpicture}[
      draw=black,line width=.8pt,
      box/.style = {rounded corners=.5},
      head/.style = {red},
      code/.style = {black!20},
      xscale=1/1.7, yscale=1/2,
            baseline=0
  ]
  \newcommand{\gate}[3]{
    \draw[fill=white,box] (#1,#2-1/2-.3) rectangle (#1+1,#2-1/2+1+.3);
    \node[align=center] at (#1+.5,#2) {#3};
  }
  \DeclareDocumentCommand{\conn}{ r() r() }{
    \draw[head] (#1) -| ($(#1) !.5! (#2)$) |- (#2);
  }  
  \node[right] at (-.25, 4) {\textsc{gate} (\texttt{G})};
  
  \foreach\x/\c in {1/0, 2/0, 3/1, 4/1, 6/0} {
    \draw[on background layer,thin,black!20] (\x,2.7) -- (\x,-1.7);
    \node at (\x,3) {\c};
  }
  \foreach\i in {-1,...,3}
     {\draw[] (.2,\i-.5) -- (6.8,\i-.5);}
  
  \draw[red,dashed] (0, 0) -- (1, 0);
  \conn (1, 0) (2, -1);
  \conn (2, -1) (3, 0);
  \conn (3, 0) (4, 1);
  \draw[red,dashed] (6, 2) -- (7, 2);
  
  \node at (1, 0) {\Sr};
  \node at (2, -1) {\Sl};
  \node at (3, 0) {\Sl};
  \node at (4, 1) {\Sl};
  \gate{4.5}{1}{$\op G$};
  \node at (6, 2) {\Sr};
  \end{tikzpicture}
  \begin{tikzpicture}[
      draw=black,line width=.8pt,
      box/.style = {rounded corners=.5},
      head/.style = {red},
      code/.style = {black!20},
      xscale=1/1.6, yscale=1/2,
            baseline=0
  ]
  \newcommand{\gate}[3]{
    \draw[fill=white,box] (#1,#2-1/2-.3) rectangle (#1+1,#2-1/2+1+.3);
    \node[align=center] at (#1+.5,#2) {#3};
  }
  \DeclareDocumentCommand{\conn}{ r() r() }{
    \draw[head] (#1) -| ($(#1) !.5! (#2)$) |- (#2);
  }  
  \node[right] at (-.25, 4) {\textsc{inv} (\texttt{I})};
  
  \foreach\x/\c in {1/0, 2/1, 4/1, 5/0, 6/0} {
    \draw[on background layer,thin,black!20] (\x,2.7) -- (\x,-1.7);
    \node at (\x,3) {\c};
  }
  \foreach\i in {-1,...,3}
     {\draw[] (.2,\i-.5) -- (6.8,\i-.5);}
  
  \draw[red,dashed] (0, 2) -- (1, 2);
  \conn (1, 2) (2, 1);
  \conn (4, 0) (5, -1);
  \conn (5, -1) (6, 0);
  \draw[red,dashed] (6, 0) -- (7, 0);
  
  \node at (1, 2) {\Sr};
  \node at (2, 1) {\Sr};
  \node at (4, 0) {\Sr};
  \node at (5, -1) {\Sl};
  \gate{2.5}{1}{$\op G^\dagger$};
  \node at (6, 0) {\Sr};
  \end{tikzpicture}
  \caption{Execution order of an arbitrary circuit approximated using the universal gate $\op G$ and its inverse $\op G^\dagger$. Each elementary operation start and end in a configuration \Sr, where the last program bit is a 0---like this, each circuit can be constructed by a simple combination of these elementary operations, with a constant overhead. Observe that both gate application and inverse gate application do not end on the same line, which means that if we want to apply $\op G$ at the current position, we have to execute \texttt{DGD}, and similarly \texttt{UIU} for $\op G^\dagger$. The specific quantum gate $\op G$ that we use is proven to be universal in \cref{lem:single-gate-universality}.}
  \label{fig:circuit-to-desc}
\end{figure}

For suitable circuits $C_{\op R}$, this construction is thus a history state Hamiltonian which encodes an arbitrary QRM.

\begin{remark}\label{rem:side-lengths}
If we want to encode a QRM which runs for $t$ applications of the QRM head, we necessarily need $H\ge 2t(D+W)$ for our cube. Furthermore, if the QRM head acts on two qudits of dimension $d$, the circuit $C_{\op R}$ acts on $m=\lceil\log_2 d\rceil$ qubits; we thus require $D+W \equiv 0 \pmod m$.
\end{remark}

For a fixed cube depth $D$ encoding some \BQEXP QRM, we need to ensure that we can tune the remaining two free parameters $W$ and $H$---width and height of the cuboid---to provide enough space and time for the computation to run and terminate, while at the same time keeping the error introduced by approximating the QRM head unitary within bounds. This is captured in the following technical lemma.

\begin{lemma}\label{lem:which-cube}
    Take a \BQEXP promise problem $\Pi$. For any precision $\delta>0$ and instance $l\in\Pi$, there exist cube parameters $W, H, D=\BigO(\exp\poly (|l|,\log 1/\delta))$ which allow a verifier ring machine to be executed on the cube for instance $l$ to within precision $\delta$.
\end{lemma}
\begin{proof}
Let $l\in\Pi$. A \BQEXP witness computation for this instance $l$ of size $|l|$ can be performed with a QRM with head unitary $\op R\in\set{SU}((\field C^d)^{\otimes 2})$ for some $d$.  We require that the QRM head $\op R$ contains a description of instance $l$; this means that $d$---the size of each of the two qudits that $\op R$ acts on---depends on the size of the instance, i.e.\ $d=\BigO(\poly |l|)$.
Denote with $t$ the number of steps the ring machine needs to perform to run the entire verifier computation.

\begin{enumerate}
\item
In \cref{lem:single-gate-universality} we show that there exists a specific 2-qubit gate $\op G$ which is universal for quantum computation, even when only applied to adjacent qubits.
\item
Using S-K and a circuit encoding as described in \cref{fig:circuit-to-desc} using gates $\op G$ and its inverse $\op G^\dagger$, approximate the QRM head $\op R$ with circuit $C_{\op R}$ to some error $\epsilon\le\delta/t$, where $\delta$ is the overall precision which we require for the verifier. Each qudit $\field C^d$ of the QRM verifier is encoded in $ m=\lceil \log_2 d \rceil$ qubits.
The circuit $C_{\op R}$ thus acts on $(\field C^2)^{\otimes 2m}$, i.e.\ $m$ qubits.
By \cite{Nielsen2010}, approximating an $n$-qubit unitary to within precision $\epsilon$ requires $\BigO(n^2 4^n log^c(n^2 4^n/\epsilon))$ gates (for some $c\le4$), if using their gateset; for our purposes it suffices to know that the number of gates required to approximate $\op R$ to within precision $\epsilon$ scales as $\BigO(\poly(d)\times \log^c(1/\epsilon))$.
\item
The circuit description is thus of length $|C_{\op R}|=\BigO(\poly |l| \times \log^c(1/\epsilon))$ and therefore we have to require that the depth of the cube $D=\BigO(\exp(|C_{\op R}|))=\BigO(\exp(\poly |l| \times \log^c(1/\epsilon)))$. 
\item
The front sidelength $W$ is increased...
\begin{itemize}
    \item to make the ring $r=W+D$ large enough for the computation, if it is not already, and
    \item to make the ring size an integer multiple of $m=\lceil \log_2 d \rceil$.
\end{itemize}
\item Set $H=2t(W+D)$.
\end{enumerate}
With $\epsilon\le\delta/t$ and $t=\BigO(\exp\poly |l|)$, we further have $\log^c 1/\epsilon\le\log^c(t/\delta)=\BigO(\poly (|l|,\log 1/\delta))$, and the claim of the lemma follows.
\end{proof}

\begin{remark}\label{rem:cube-accuracy}
If we require cube parameters of $\BigO(\exp\poly |l|)$, we can demand a computation accuracy of at most $\delta=\Omega(1/\exp\poly|l|)$.
\end{remark}
\begin{proof}
If we demand the two scaling parameters in \cref{lem:which-cube} to be equal, we have
\begin{alignat*}{2}
    &\quad                 &\exp \log^4\left(1/\delta\right) &= \BigO(\exp \poly |l|)\\
    &\Leftrightarrow\quad  &\log^4\left(1/\delta\right)      &= \BigO(\poly |l|)\\
    &\Leftrightarrow\quad  &\log\left(1/\delta \right)       &= \BigO(\poly |l|)\\
    &\Leftrightarrow\quad  &\delta                           &= \Omega(1/\exp \poly |l|).
    \tag*{\qedhere}
\end{alignat*}
\end{proof}

\subsection{Static Lattice Constraints}\label{sec:static-constraints}

\subsubsection{Lattice Structure}
We will work with a face-centered cubic lattice of 4-dimensional qudits. All interactions will be at most 4-local and translationally-invariant. The system will have open boundary conditions; in particular, we do not cut off interactions at the boundary or introduce boundary constraints of any kind.
For the sake of clarity, when writing out constraints in the following, we will usually ignore parts of the sublattice, implicitly assuming that any interaction term is extended trivially everywhere else. When refering to layer A and if not explicitly mentioned, we mean the black sublattice, and layer B will be the red sublattice with side-centered vertices.

Any ``static'' constraint---i.e.\ the terms in the following four subsections---will be translated into local Hamiltonian terms diagonal in the computational basis; see \cref{sec:summary-static} for details.

\subsubsection{Constraining the Lattice Bulk}\label{sec:constraining-bulk}

Denoting with \inline{\Rx} a special symbol in the red sublattice, we want to constrain the lattice to be in this state in the bulk, and in its complement on the topmost red face, as well as the outermost side faces. We first give a bonus of 1 to spins in the B sublattice in configuration
\[
\tikz[lattice]{
    \draw[vertex=red,faint] (0, 0, 0) -- (0, 0, -1);
    \Rx;
    \R{0}{0}{-1};
}
\]
All red layers but the top one will then be in state \inline\Rx. We then give a bonus of 1 to all of the following configurations:
\[
\tikz[lattice]{
    \draw[vertex=red,faint] (0, 0, 0) -- (0, 0, -1);
    \draw[vertex=red,faint] (0, 0, 0) -- (0, 1, 0);
    \draw[vertex=red,faint] (0, 0, 0) -- (1, 0, 0);
    \Rx;
    \R{0}{0}{-1};
    \R{0}{1}{0};
    \R{1}{0}{0};
}
\quad
\tikz[lattice]{
    \draw[vertex=red,faint] (0, 0, 0) -- (0, 0, -1);
    \draw[vertex=red,faint] (0, 0, 0) -- (0, 1, 0);
    \draw[vertex=red,faint] (0, 0, 0) -- (-1, 0, 0);
    \Rx;
    \R{0}{0}{-1};
    \R{0}{1}{0};
    \R{-1}{0}{0};
}
\quad
\tikz[lattice]{
    \draw[vertex=red,faint] (0, 0, 0) -- (0, 0, -1);
    \draw[vertex=red,faint] (0, 0, 0) -- (0, -1, 0);
    \draw[vertex=red,faint] (0, 0, 0) -- (1, 0, 0);
    \Rx;
    \R{0}{0}{-1};
    \R{0}{-1}{0};
    \R{1}{0}{0};
}
\quad
\tikz[lattice]{
    \draw[vertex=red,faint] (0, 0, 0) -- (0, 0, -1);
    \draw[vertex=red,faint] (0, 0, 0) -- (0, -1, 0);
    \draw[vertex=red,faint] (0, 0, 0) -- (-1, 0, 0);
    \Rx;
    \R{0}{0}{-1};
    \R{0}{-1}{0};
    \R{-1}{0}{0};
}
\quad
\]
This leaves the top layer unchanged. Summarizing the bonus terms so far, all other B layers, as seen from the top, are then in the configuration
\[
\begin{tikzpicture}[
    lattice
]
    \foreach \x in {-6,...,-2} {
        \draw (\x,1,0) -- (\x,4.2,0);
        \foreach \y in {1,2,3} {
            \Rx[opacity=1][\x][\y+.5][0];
            \node[xshift=-6,yshift=3] at (\x, \y+.5, 0) {\footnotesize
                \ifnumcomp{\x}{=}{-6}{
                    2
                }{
                    \textbf 5
                }
            };
        }
    }
    \foreach \y in {1,2,3,4}
        \draw (-6,\y,0) -- (-1.8,\y,0);
    \foreach \x in {-6,...,-3}
        \foreach \y in {1,2,3,4} {
            \Rx[opacity=1][\x+.5][\y][0];
            \node[xshift=6,yshift=4] at (\x+.5, \y, 0) {\footnotesize
                \ifnumcomp{\y}{=}{1}{
                    2
                }{
                    \textbf 5
                }
            };
       }

\end{tikzpicture}
\]
We then give a global 1-local penalty to \inline\Rx with strength -3. The top B layer will thus be in the complement of \inline\Rx (which we denote with \inline\Ro), while all the other B layers look like
\[
\begin{tikzpicture}[
    lattice
]
    \foreach \x in {-6,...,-2} {
        \draw (\x,1,0) -- (\x,4.2,0);
        \foreach \y in {1,2,3} {
            \ifnumcomp{\x}{=}{-6}{
                \Ro[\x][\y+.5][0];
            }{
                \Rx[opacity=1][\x][\y+.5][0];
                \node[xshift=-6,yshift=3] at (\x, \y+.5, 0) {\footnotesize 2};
            }
        }
    }
    \foreach \y in {1,2,3,4}
        \draw (-6,\y,0) -- (-1.8,\y,0);
    \foreach \x in {-6,...,-3}
        \foreach \y in {1,2,3,4} {
            \ifnumcomp{\y}{=}{1}{
                \Ro[\x+.5][\y][0];
            }{
                \Rx[opacity=1][\x+.5][\y][0];
                \node[xshift=6,yshift=4] at (\x+.5, \y, 0) {\footnotesize 2};
            }
       }

\end{tikzpicture}
\]

\subsubsection{Binary Counter}
The top layer of type B will carry a binary counter tiling, which translates the side length $D$ into a binary representation on the top front edge. In order to achieve this, we need to initialize the top back edge of the cube to all 0, and the top right edge to all 1. Since we do not want to use distinct interactions on the outside layers, but have open boundary conditions, we have to find a configuration in the pre-constrained cube which only occurs on the top right and back edge, respectively. The following configuration is such an example:
\begin{equation}\label{eq:top-side-interaction}
\begin{tikzpicture}[    
	lattice
]    
    \foreach \x/\y in {0/1,0/2}
        \draw (\x,\y,0) -- (\x,\y,1.7);
    
    \foreach \y/\z in {1/0,2/0,1/1}
        \draw (0,\y,\z) -- (0.8,\y,\z);
        
    \foreach \x/\z in {0/0,0/1}
        \draw (\x,0.8,\z) -- (\x,2.8,\z);
        
    \foreach \x in {0}
        \foreach \y in {1,2}
            \foreach \z in {0,1}
                \B \x \y \z;
        
    \draw[dashed] (0, 0, .5) -- (0, 3, .5);
    \draw[draw=none, fill=black, fill opacity=.2] (1, 0, .5) -- (0, 0, .5) -- (0, 3, .5) -- (1, 3, .5);
    \draw[<-] (.7, .5, .5) -- (3, 1, .5) node[anchor=east] {top B layer};
    
    \draw[intA] (0, 2.5, 1.5) -- (0, 1.5, 0.5) -- (.5, 1, 0.5) (0, 1.5, .5) -- (0, 1.5, 1.5);

    \foreach \x/\y/\z in {0/1.5/.5, 0/1.5/1.5, 0/2.5/.5, 0/2.5/1.5, .5/1/.5, .5/2/.5}
        \Ro[\x][\y][\z];
    \foreach \x/\y/\z in {.5/1/1.5, .5/2/1.5}
        \Rx[][\x][\y][\z];
\end{tikzpicture}
\end{equation}

Since only the top and outer layers have red spins in configuration \inline{\Ro}, this four-local interaction allows us to pick out the top right boundary of the top layer, and to constrain it to state 1. A similar interaction allows constraining the top back layer to 0. The top B layer then looks like
\[
\begin{tikzpicture}[
    lattice,
]
    \foreach \x in {-6,...,-2} {
        \draw (\x,0.8,0) -- (\x,4,0);
        \foreach \y in {1,2,3} {
            \ifnumcomp{\x}{=}{-6}{
                \draw[vertex=tileRed, draw=none] (\x, \y, 0) -- (\x, \y+1, 0) -- (\x+.5, \y+.5, 0) -- cycle;
                \node[black,align=center] at (\x, \y+.5, 0) {\$1};
            }{
                \Ro[\x][\y+.5][0];
            }
        }
    }
    \foreach \y in {1,2,3,4}
        \draw (-6,\y,0) -- (-1.8,\y,0);
    \foreach \x in {-6,...,-3}
        \foreach \y in {1,2,3,4} {
            \ifnumcomp{\y}{=}{4}{
                \node[black,align=center] at (\x+.6, \y, 0) {\$0};
            }{
                \Ro[\x+.5][\y][0];
            }
       }
\end{tikzpicture}
\]
Since this is the only B layer with spins in state \inline\Ro, we can use the following tiles from \cite{Patitz2014} to get the desired binary counting layer.
\[
\tile0000
\quad
\tile1101
\quad
\tile1010
\quad
\tile0110
\]
It is straightforward to verify that the general tile
\[
\tile absc
\]
obeys the rules $c=\text{carry of }a+b$ and sum $s=a\oplus_2b$.

\subsubsection{Winding Program Diagonally}
We use an interaction similar to \cref{eq:top-side-interaction} to shuffle the program around the cube in a cyclic fashion, as depicted in \cref{fig:cube-structure}:
\DeclareDocumentCommand{\shufFig}{ m }{%
\begin{tikzpicture}[    
	lattice
]
    \foreach \x/\y in {0/1,0/2}
        \draw (\x,\y,0) -- (\x,\y,1.7);
    
    \foreach \y/\z in {1/0,2/0,1/1}
        \draw (0,\y,\z) -- (0.8,\y,\z);
        
    \foreach \x/\z in {0/0,0/1}
        \draw (\x,0.8,\z) -- (\x,2.8,\z);
        
    \foreach \x in {0}
        \foreach \y in {1,2}
            \foreach \z in {0,1}
                \B \x \y \z;
    
    \draw[intA] (0, 2.5, 1.5) -- (0, 1.5, 0.5) -- (0.5, 2, 0.5);
    
    \node[red] at (0, 1.5, .5) {\${#1}};
    \node[red] at (0, 2.5, 1.5) {\${#1}};

    \foreach \x/\y/\z in {0/2.5/.5, 0/1.5/1.5}
        \Ro[\x][\y][\z];
    \foreach \x/\y/\z in {.5/1/.5, .5/2/.5}
        \R{\x}{\y}{\z};
    \foreach \x/\y/\z in {.5/1/1.5, .5/2/1.5}
        \Rx[][\x][\y][\z];
\end{tikzpicture}%
}
\[
    \shufFig0 \quad\text{and}\quad \shufFig1
\]
Observe that, by including the red qudit one layer in, this interaction does indeed only apply to the front right face; similar interactions on the other three faces achieve the desired program copying around the cube sides. Additionally, by conditioning on if this inner qudit is either  \inline{\Rx} or \inline{\Ro}, we can apply a different rule at the top layer. In particular at the top layer of the front right face we want to flip the bit when copying down so that there are 1's on the top layer but 0's on the layer below - see \cref{fig:cube-structure}.

On the corners, we use a similar shape of interaction, i.e.
\[
\begin{tikzpicture}[    
	lattice
]
    \foreach \x/\y in {1/0,0/0,0/1}
        \draw (\x,\y,-.3) -- (\x,\y,1.7);
    
    \foreach \y/\z in {0/0,1/0,0/1}
        \draw (0,\y,\z) -- (1.7,\y,\z);
        
    \foreach \x/\z in {0/0,1/0,0/1}
        \draw (\x,0,\z) -- (\x,1.7,\z);
        
    \foreach \x in {0,1}
        \foreach \y in {0,1}
            \foreach \z in {0,1}
                \B \x \y \z;
    
    \draw[intA] (1,.5,.5) -- (.5, 0, .5) -- (0,.5,1.5);

    \Rx[][1][.5][.5];
    \Rx[][.5][1][.5];
    \Rx[][1][.5][1.5];
    \Rx[][.5][1][1.5];
    \node[red] at (1.5, 0, .5) {\${$p_{i+1}$}};
    \node[red] at (.5, 0, .5) {\${$p_{i}$}};
    \node[red] at (0, .5, .5) {\${$p_{i-1}$}};
    \node[red] at (0, 1.5, .5) {\${$p_{i+1}$}};
    \node[red] at (1.5, 0, 1.5) {\${$p_{i+2}$}};
    \node[red] at (.5, 0, 1.5) {\${$p_{i+1}$}};
    \node[red] at (0, .5, 1.5) {\${$p_{i}$}};
    \node[red] at (0, 1.5, 1.5) {\${$p_{i-1}$}};
\end{tikzpicture}    
\]
and similarly for all other corners.

Note that in \cref{sec:comp}, we will need to temporarily replace red program bits with a special symbol \textcolor{red}! indicating that the application of a gate is happening in the next step, so we exclude this case from the constraints in this section (i.e.\ we allow \emph{either} $p_i$ and $p_i$, or $p_i$ and \textcolor{red}! to appear around the computational corner, and similarly for the diagonal face constraints bordering the computation edge).

As none of the dynamic transition rules below ever changes the number of head symbols (of which \textcolor{red}! is one), we can rule out the cases where there is more than one \textcolor{red}! or other head symbol present at any one time---we analyse these branching cases in detail in \cref{sec:branching}.

\subsubsection{Constraining layer A qudits}

We label the states of the green face-centred qubits of the layer A type with the alphabet $\{A,B,C,0\}$. 
For all such green lattice qubits we apply a bonus of strength $1/2$ to configuration 0, so that this state is preferred.

In order to access two sequential program bits $p_i$ and $p_{i+1}$ with a single three-local interaction on the computation edge, we add a strength 1 interaction which constrains the front column of the layer A green sublattice to a state $P_i \in \{A,B,C\}$ depending on the two neighbouring computation bits, i.e.
\[
\begin{tikzpicture}[    
	lattice
]
    \foreach \x/\y in {1/0,0/0,0/1}
        \draw (\x,\y,-.3) -- (\x,\y,1.7);
    
    \foreach \y/\z in {0/0,1/0,0/1}
        \draw (0,\y,\z) -- (1.7,\y,\z);
        
    \foreach \x/\z in {0/0,1/0,0/1}
        \draw (\x,0,\z) -- (\x,1.7,\z);
        
    \foreach \x in {0,1}
        \foreach \y in {0,1}
            \foreach \z in {0,1}
                \B \x \y \z;
    
    \draw[intA] (.5, 0, .5) -- (.5, .5, 0) -- (0, .5, .5);

    \node[green] at (.5, .5, 0) {\${$P_i$}};
    \node[red] at (1.5, 0, .5) {\${$p_{i+2}$}};
    \node[red] at (.5, 0, .5) {\${$p_{i+1}$}};
    \node[red] at (0, .5, .5) {\${$p_{i}$}};
    \node[red] at (0, 1.5, .5) {\${$p_{i-1}$}};
\end{tikzpicture}    
\]
Note that this interaction will have no effect anywhere else in the lattice, as at least one of the two red program bits will be \inline\Rx.

The rule which governs what state $P_i$ is constrained to will depend on the tuple $p_i$ and $p_{i+1}$, and is derived from \cref{tab:encoding}. The idea is that $P_i=f(p_i,p_{i+1})$ will signify what is to happen at the computation edge. Looking at \cref{tab:encoding}, we see that at each stage we either:
\begin{itemize}
\item[A.] \textbf{A}pply a gate (either $\op G$ or its inverse $\op G^\dagger$, depending on where the arrow is coming from), 
\item[B.] go \textbf{B}ackwards (i.e. change the direction of the arrow),
\item[C.] or \textbf{C}ontinue in the same direction.
\end{itemize}
Given the encoding of \cref{tab:encoding}, we therefore take $P_i=f(p_i,p_{i+1})$ for a function $f$ given by
$$
    f(p_i,p_{i+1})=\begin{cases}
        B & \text{if $p_{i+1}=0$,} \\
        C & \text{if $p_i=0$ and $p_{i+1}=1$, and} \\
        A & \text{if $p_i=p_{i+1}=1$.}
    \end{cases}
$$
Due to the aforementioned $1/2$ bonus which applies at all green spins, the remainder of layer A is in configuration $0$.

\subsubsection{Summary of static constraints}\label{sec:summary-static}

As explained in the main text under ``Tiling Construction'', we take all static constraints listed so far and translate them to diagonal and local projectors $\op h_i$.
This allows us to write a 4-local, translationally-invariant classical Hamiltonian $\Hstat=\sum_{\vec x}\sum_{\op h_i}\op h_i^{\vec x}$ (i.e.\ product and diagonal in the computational basis of each spin) with a ground space spanned by states with the following properties.
\begin{enumerate}
\item Any black vertex spin in layer A is unconstrained.
\item The red layer B spins will be in a state as depicted in \cref{fig:cube-structure}, i.e.\ on the top cuboid face, they represent a binary counter translating the depth $D$ of the cuboid into a binary description of $D$ on the top front edge. This binary string $s=p_1\ldots p_T$ is wound down diagonally around the cube, which expresses $s$ periodically on the front computation edge. \emph{Only} the spins adjacent to this edge are also allowed in a configuration \textcolor{red}!.  In the bulk of the cube all the way to the bottom-most layer, the red spins are in state \inline\Rx.
\item The green layer A is in configuration 0 everywhere but on the front edge; there, the spins there are in a configuration depending on the two adjacent program bits $p_i$ and $p_{i+1}$, as outlined above.
\end{enumerate}
This Hamiltonian $\Hstat$ is gapped with a size-independent constant gap, and we can rescale the interactions so far and shift the overall energy to assume that this ground space as detailed above has energy zero, and any other configuration has energy lower-bounded by 1.

In the next sections, we will explain the history state construction, which---within this ground space of $\Hstat$---will represent a valid QRM evolution for the circuit represented by the binary string $s$.

\subsection{Dynamic Constraints on Computational Layer}\label{sec:dynamic-constraints}
The ``dynamic'' history state transition rules will be translated in a similar fashion to terms as in \cref{eq:history-state-ham,eq:transition-term}.
We always depict a transition rule as connected by a squiggly arrow $\leadsto$; the notation is self-explanatory: the brighter blue shading indicates the original state, whereas the dull blue shading indicates the target configuration.
To give an example, a transition
$$
    \begin{tikzpicture}[lattice]
        \draw (0, -.3, 0) -- (0, 3.3, 0);
        \draw[intA] (0, 1, 0) -- (0, 2, 0);
        \node at (0, 0, 0) {\${$a$}};
        \node at (0, 1, 0) {\${$b$}};
        \node at (0, 2, 0) {\${$c$}};
        \node at (0, 3, 0) {\${$d$}};
    \end{tikzpicture}
    \raisebox{0.5cm}{$\leadsto$}\quad
    \begin{tikzpicture}[lattice]
        \draw (0, -.3, 0) -- (0, 3.3, 0);
        \draw[intB] (0, 1, 0) -- (0, 2, 0);
        \node at (0, 0, 0) {\${$a$}};
        \node at (0, 1, 0) {\${$b'$}};
        \node at (0, 2, 0) {\${$c'$}};
        \node at (0, 3, 0) {\${$d$}};
    \end{tikzpicture}
$$
would be translated into a two-local term $\op h=\ketbra{bc}+\ketbra{b'c'} - \ketbra{b'c'}{bc} - \ketbra{bc}{b'c'}$, and correspondingly with an extra quantum register if $b$ or $c$ were labelling vertices that carry a qubit (i.e.\ the black layer A sublattice vertices).

\subsubsection{Moving Qubits}\label{sec:moving-constraints}
The black sublattice (A layers) comprises the alphabet $\{0, 1, \Sr, \Sl\}$, where we treat the $0,1$-subspace as a qubit, i.e. $\field C^2$. The right and left arrows are markers to indicate where to move qubits to. As an example on the front face, we have a left moving sequence
\DeclareDocumentCommand{\movingFig}{ m m m m m m m m g }{%
\begin{tikzpicture}[    
	lattice
]
    \foreach \x in {.5, 1.5, 2.5} \Ro[\x][0][.5];
    \foreach \x in {0, ..., 3} \draw (\x, 0, -.2) -- (\x, 0, 1.2);
    \foreach \z in {0, 1} \draw (-.2, 0, \z) -- (3.2, 0, \z);
    
    #9;
    
    \node at (3, 0, 0) {\${#1}};
    \node at (2, 0, 0) {\${#2}};
    \node at (1, 0, 0) {\${#3}};
    \node at (0, 0, 0) {\${#4}};
    \node at (3, 0, 1) {\${#5}};
    \node at (2, 0, 1) {\${#6}};
    \node at (1, 0, 1) {\${#7}};
    \node at (0, 0, 1) {\${#8}};
\end{tikzpicture}%
}
\[
    \movingFig{a}{b}{c}{\Sl}{x}{y}{z}{}{
        \draw[intA] (2, 0, 1) -- (1, 0, 0) -- (0, 0, 0);
    }
    \quad\leadsto\quad
    \movingFig{a}{b}{\Sl}{y}{x}{c}{z}{}{
        \draw[intB] (2, 0, 1) -- (1, 0, 0) -- (0, 0, 0);
        \draw[intA] (3, 0, 1) -- (2, 0, 0) -- (1, 0, 0);
    }
    \quad\leadsto\quad
    \movingFig{a}{\Sl}{x}{y}{b}{c}{z}{}{
        \draw[intB] (3, 0, 1) -- (2, 0, 0) -- (1, 0, 0);
    }
\]
and analogously the right moving sequence
\[
    \movingFig{\Sr}{a}{b}{c}{}{x}{y}{z}{
        \draw[intA] (3, 0, 0) -- (2, 0, 0) -- (1, 0, 1);
    }
    \quad\leadsto\quad
    \movingFig{y}{\Sr}{b}{c}{}{x}{a}{z}{
        \draw[intB] (3, 0, 0) -- (2, 0, 0) -- (1, 0, 1);
        \draw[intA] (2, 0, 0) -- (1, 0, 0) -- (0, 0, 1);
    }
    \quad\leadsto\quad
    \movingFig{y}{z}{\Sr}{c}{}{x}{a}{b}{
        \draw[intB] (2, 0, 0) -- (1, 0, 0) -- (0, 0, 1);
    }
\]

To move qubits around a corner, we use an interaction of the form
\DeclareDocumentCommand{\movingCornerFig}{ m m m m m m m m g }{%
\begin{tikzpicture}[    
	lattice
]
    \foreach \x/\y in {1.5/0, .5/0, 0/.5, 0/1.5} \Ro[\x][\y][.5];
    \foreach \x/\y in {1/0, 0/0, 0/1, 0/2} \draw (\x, \y, -.2) -- (\x, \y, 1.2);
    \foreach \z in {0, 1} \draw (1.8, 0, \z) -- (0, 0, \z) -- (0, 2.2, \z);
    
    #9;
    
    \node at (1, 0, 0) {\${#1}};
    \node at (0, 0, 0) {\${#2}};
    \node at (0, 1, 0) {\${#3}};
    \node at (0, 2, 0) {\${#4}};
    \node at (1, 0, 1) {\${#5}};
    \node at (0, 0, 1) {\${#6}};
    \node at (0, 1, 1) {\${#7}};
    \node at (0, 2, 1) {\${#8}};
\end{tikzpicture}%
}
\[
\movingCornerFig{\Sr}{a}{b}{c}{}{x}{y}{z}{
    \draw[intA] (1, 0, 0) -- (0, 0, 0) -- (0, 1, 1);
}
\quad\leadsto\quad
\movingCornerFig{y}{$\Sr^*$}{b}{c}{}{x}{a}{z}{
    \draw[intB] (1, 0, 0) -- (0, 0, 0) -- (0, 1, 1);
}
\]
at the back, left and right corners (different rules as described in \cref{sec:comp} are used for the front edge) and similarly for going around the corner in the opposite direction.

A few remarks: first note that all the transitions defined so far are unique, i.e. given the cube bulk constrained to \inline\Rx as done in \cref{sec:constraining-bulk}, and for every configuration with only one arrow symbol (the other cases we will penalize as a last step), there exists precisely one forward and one backwards transition. Another important point is how to modify the arrows when going around the circumference of the cube once (marked with a $\Sr^*$ in the last transition rule); at the moment, if we left the arrow type unchanged for every corner, we would not be able to shuffle around the qubits in a circle; on the back face, we would be doing the opposite shuffling operation. Therefore, we change the arrow type according to the following scheme:
\DeclareDocumentCommand{\cubeArrowFig}{ m m m O{} }{%
\begin{tikzpicture}[
    lattice
]

\draw[faint] (3, 3, 0) -- (3, 3, 1.2);
\foreach \x in {0.05, 0.1, ..., 1} {
    \path[
        decoration = {
            markings,
            mark = at position \x with {
                \node[faint] {\${#2}};
            }
        },
        postaction = decorate,
        draw = none
    ] (3, 0, 1) -- (3, 3, 1) -- (0, 3, 1);
}

\draw (0, 0, 0) -- (3, 0, 0) -- (3, 3, 0) -- (0, 3, 0) -- cycle;
\draw (0, 0, 0) -- (0, 0, 1.2);
\draw (3, 0, 0) -- (3, 0, 1.2);
\draw (0, 3, 0) -- (0, 3, 1.2);

\foreach \x in {0.05, 0.1, ..., 1} {
    \path[
        decoration = {
            markings,
            mark = at position \x with {
                \node {\${#1}};
            }
        },
        postaction = decorate,
        draw = none
    ] (3, 0, 1) -- (0, 0, 1) -- (0, 3, 1);
}

\draw[red, ->] (0, 0, 0.8) -- (#3, 0.8);
\draw[
    draw = none,
    postaction = {
        decorate,
        decoration = {
            text along path, #4,
            text = {start},
            text align = {center},
            text color = red
        }
    },
    yshift = 1.5
] (0, 0, 0.8) -- (#3, 0.8);
\end{tikzpicture}%
}
\[
\cubeArrowFig{\Sr}{\Sl}{0, 1}
\quad\text{and}\quad
\cubeArrowFig{\Sl}{\Sr}{1, 0}[reverse path]
\]

\subsubsection{Computation}\label{sec:comp}
\begin{table}

\centering
\begin{tabular}{ccccc}
	\toprule
	tape coming from & program $x_{n-1}x_n$ & operation on qubits & tape going to &  \\ \midrule
	      \Sr        &          00          &        $\1$         &      \Sl      &  \\
	      \Sr        &          01          &        $\1$         &      \Sr      &  \\
	      \Sr        &          10          &        $\1$         &      \Sl      &  \\
	      \Sr        &          11          &   $\op G^\dagger$   &      \Sr      &  \\ \hline
	      \Sl        &          00          &        $\1$         &      \Sr      &  \\
	      \Sl        &          01          &        $\1$         &      \Sl      &  \\
	      \Sl        &          10          &        $\1$         &      \Sr      &  \\
	      \Sl        &          11          &       $\op G$       &      \Sl      &  \\ \bottomrule
\end{tabular}
\caption{Program encoding. The arrow symbols \Sr and \Sl indicate in which direction the ring is moving. Relative to the tape, the current head is thus moving in the opposite direction. With this encoding, any circuit can be executed with the available operations, cf.\ \cref{fig:circuit-to-desc}.}
\label{tab:encoding}
\end{table}

In order to execute any circuit as in \cref{fig:circuit-to-desc}, we have eight elementary operations available, all of which are listed in \cref{tab:encoding}.
It is easy to see that there exists a symmetry between the right- and left-moving arrow; we will thus explain the right-moving arrows (including the application of gate $\op G$) in detail and leave the reverse direction as an exercise to the reader.

\NewDocumentCommand{\computationFig}{ m }{%
\begin{tikzpicture}[
    lattice,
    xscale=.8,
    yscale=.8
]

\foreach \z in {.5, 1.5}
    \foreach \x/\y in {2.5/0, 1.5/0, .5/0, 0/.5, 0/1.5, 0/2.5}
        \Ro[\x][\y][\z];
\foreach \x/\y in {3/0, 2/0, 1/0, 0/0, 0/1, 0/2, 0/3} \draw (\x, \y, -.2) -- (\x, \y, 2.2);
\foreach \z in {0, 1, 2} \draw (3.2, 0, \z) -- (0, 0, \z) -- (0, 3.2, \z);

#1
\end{tikzpicture}%
}
\NewDocumentCommand{\computationBlackVertFig}{ m O{?} O{?} O{?} O{?} O{?} O{?} O{?} }{%
\node at (3, 0, #1) {\${#2}};
\node at (2, 0, #1) {\${#3}};
\node at (1, 0, #1) {\${#4}};
\node at (0, 0, #1) {\${#5}};
\node at (0, 1, #1) {\${#6}};
\node at (0, 2, #1) {\${#7}};
\node at (0, 3, #1) {\${#8}};
}
\NewDocumentCommand{\computationRedVertFig}{ m O{} O{} O{} O{} O{} O{} }{%
\node[red] at (2.5, 0, #1) {\${#2}};
\node[red] at (1.5, 0, #1) {\${#3}};
\node[red] at (0.5, 0, #1) {\${#4}};
\node[red] at (0, 0.5, #1) {\${#5}};
\node[red] at (0, 1.5, #1) {\${#6}};
\node[red] at (0, 2.5, #1) {\${#7}};
}
\DeclareRobustCommand{\funcF}[2]{
\ifthenelse{\equal{#2}{0}}{B}{
\ifthenelse{\equal{#1}{0}}{C}{A}}
}

\NewDocumentCommand{\computationRedVertDefault}{ m m }{%
\node[green] at (0.5, 0.5, 0) {\${$\funcF{#1}{#2}$}};
\computationRedVertFig{0.5}[$p_{4}$][$p_3$][$#2$][$#1$][$p_n$][$p_{n-1}$]
\computationRedVertFig{1.5}[$p_5$][$p_{4}$][$p_3$][$#2$][$#1$][$p_n$]
}

\NewDocumentCommand{\intA}{ r() }{ \fill[draw=none, fill=Cerulean] (#1) circle [radius=8pt]; }
\NewDocumentCommand{\intB}{ r() }{ \fill[draw=none, pattern color=MidnightBlue, pattern=crosshatch dots] (#1) circle [radius=8pt]; }

\noindent
\begin{tabular}[\textwidth]{ p{8.5cm} p{\textwidth-10cm} }
\computationFig{
    \intA (1, 0, 0);
    \intA (0.5, 0.5, 0);
    \intA (0, 0, 1);
    \intA (0, 1, 1);
    \computationRedVertDefault 00
    \computationBlackVertFig{0}[?][?][\Sr]
    \computationBlackVertFig{1}[y][z][a][b][?'][c]
    \computationBlackVertFig{2}
}
&
Consider $p_1p_2=00$ or $10$, so that $f(p_1,p_2)=B$.
The transition is conditioned to only happen if the green qubit in the layer A sublattice is in the $B$ state.
\\[-1cm]
\computationFig{
    \intB (1, 0, 0);
    \intB (0.5, 0.5, 0);
    \intB (0, 0, 1);
    \intB (0, 1, 1);
    \computationRedVertDefault 00
    \computationBlackVertFig{0}[?][?][?']
    \computationBlackVertFig{1}[y][z][a][\Sl][b][c]
    \computationBlackVertFig{2}
}
&
Move ?' up, b to the right and flip the arrow.
This corresponds to simply reverting the direction as in \cref{tab:encoding}.
\end{tabular}

\noindent
\begin{tabular}[\textwidth]{ p{8.5cm} p{\textwidth-10cm} }
\computationFig{
    \intA (1, 0, 0);
    \intA (0.5, 0.5, 0);
    \intA (0, 0, 1);
    \intA (0, 1, 1);
    \computationRedVertDefault 01
    \computationBlackVertFig{0}[?][?][\Sr]
    \computationBlackVertFig{1}[y][z][a][b][?'][c]
    \computationBlackVertFig{2}
}
&
Consider $p_1p_2=01$ so that $f(p_1,p_2)=C$. We perform the same action as above, but keep the arrow direction.
\\[-1cm]
\computationFig{
    \intB (1, 0, 0);
    \intB (0.5, 0.5, 0);
    \intB (0, 0, 1);
    \intB (0, 1, 1);
    \computationRedVertDefault 01
    \computationBlackVertFig{0}[?][?][?']
    \computationBlackVertFig{1}[y][z][a][\Sr][b][c]
    \computationBlackVertFig{2}
}
\end{tabular}

\noindent
\begin{tabular}[\textwidth]{ p{8.5cm} p{\textwidth-10cm} }
\computationFig{
    \intA (1, 0, 0);
    \intA (0.5, 0.5, 0);
    \intA (0, .5, .5);
    \intA (0, 1, 1);
    \computationRedVertDefault 11
    \computationBlackVertFig{0}[?][?][\Sr]
    \computationBlackVertFig{1}[y][z][a][b][?'][c]
    \computationBlackVertFig{2}
}
&
Consider $p_1p_2=11$ so that $f(p_1,p_2)=A$. We want to execute a gate, which requires one intermediate step.
\\[-1cm]
\computationFig{
    \intA (1, 0, 1);
    \intA (0, 0, 1);
    \intA (0, 1, 1);
    \intA (0, .5, .5);
    \intB (1, 0, 0);
    \intB (0.5, 0.5, 0);
    \intB (0, .5, .5);
    \intB (0, 1, 1);
    \node[green] at (0.5, 0.5, 0) {\${$A$}};
    \computationRedVertFig{0.5}[$p_{n-1}$][$p_n$][$1$][!][$p_3$][$p_4$]
    \computationRedVertFig{1.5}[$p_{n-2}$][$p_{n-1}$][$p_n$][$1$][$1$][$p_3$]
    \computationBlackVertFig{0}[?][?][?']
    \computationBlackVertFig{1}[y][z][a][b][1][c]
    \computationBlackVertFig{2}
}
&
We place the computation marker on the right hand side of the computation edge.
This signals that the next step is to perform a gate $\op G$ on a and b.
\\[-1cm]
\computationFig{
    \intB (1, 0, 1);
    \intB (0, 0, 1);
    \intB (0, 1, 1);
    \intB (0, .5, .5);
    \computationRedVertDefault 11
    \computationBlackVertFig{0}[?][?][?']
    \computationBlackVertFig{1}[y][z][a'][\Sr][b'][c]
    \computationBlackVertFig{2}
}
&
Here $\ket{a'}\ket{b'}:=\op G\ket a\ket b$. The program is restored and the arrow left in the right moving configuration, as required by \cref{tab:encoding}.
\end{tabular}

\noindent
We now move the arrow once around the tape and then arrive at the computational corner from the other side.
Observe---as mentioned---that the encoding in \cref{tab:encoding} is mirror-symmetric, so by reversing all the rules above one can implement the same rules---while applying $\op G^{-1}$ instead of $\op G$ when $P_i=A$ for an arrow incoming from the right.

\subsubsection{Computational Input and Output Constraints}
Since the instance is specified within the QRM head, it suffices to provide the computation with a single ancilla $\ket 0$ as input; in case we need more ancillas than available on the front edge, we can augment our verifier as in \cite[fig. 4]{Bausch2016}.
Due to the configuration of the red layer B sublattice, it is straightforward to find a local configuration which only ever appears on a top right corner; more specifically, we utilize the constraint interaction
\[
\begin{tikzpicture}[    
	lattice
]
    \foreach \x/\y in {0/0, 0/1, 1/0, 1/1}
        \draw (\x, \y, 0) -- (\x, \y, 1.6);
    \foreach \z in {0, 1} {
        \draw (1.2, 0, \z) -- (0, 0, \z) -- (0, 1.2, \z) (0, 1, \z) -- (1, 1, \z) -- (1, 0, \z) (1.2, 1, \z) -- (1, 1, \z) -- (1, 1.2, \z);
        \B00{\z};
        \B01{\z};
        \B10{\z};
        \B11{\z};
    }
    
    \draw[intA] (0, 0, 0) -- (.5, 0, 1.5) (0, 0, 0) -- (0, .5, 1.5) (0, 0, 0) -- (.5, 1, .5);
    
    \Ro[0][.5][.5];
    \Ro[.5][0][.5];
    \Ro[0][.5][1.5];
    \Ro[.5][0][1.5];
    \Ro[1][.5][.5];
    \Ro[.5][1][.5];
    \Rx[][1][.5][1.5];
    \Rx[][.5][1][1.5];
    \
\end{tikzpicture}
\]
to enforce that the black symbol is either in an arrow configuration, or 0, respectively. The rest of the tape is left unconstrained.

Since there is nothing special about the bottom-most layers A and B, we need to use a pair of interactions to enforce the last black qubit to an accepting state.
This can be readily achieved using
\DeclareDocumentCommand{\figOutPenalty}{ m }{
\begin{tikzpicture}[    
	lattice
]
    \foreach \x/\y in {1/0,0/0,0/1}
        \draw (\x,\y,-.3) -- (\x,\y,1.7);
    
    \foreach \y/\z in {0/0,1/0,0/1}
        \draw (0,\y,\z) -- (1.7,\y,\z);
        
    \foreach \x/\z in {0/0,1/0,0/1}
        \draw (\x,0,\z) -- (\x,1.7,\z);
        
    \foreach \x in {0,1}
        \foreach \y in {0,1}
            \foreach \z in {0,1}
                \B \x \y \z;
    
    \foreach \z in {.5,1.5} {
        \Ro[1.5][0][\z];
        \Ro[.5][0][\z];
        \Ro[0][.5][\z];
        \Ro[0][1.5][\z];
        \Rx[][.5][1][\z];
        \Rx[][1][.5][\z];
    };
    
    #1
\end{tikzpicture}
}
\begin{align*}
\figOutPenalty{
    \draw[intA] (.5,1,.49) -- (.5,1,.51) (1,.5,.49) -- (1,.5,.51) (0,.5,.49) -- (0,.5,.51) (0,0,1) -- (0,.01,1);
}
\quad&\text{constraining the black qubit to state $\ket0$, and}\\
\figOutPenalty{
    \draw[intA] (.5,1,1.49) -- (.5,1,1.51) (1,.5,1.49) -- (1,.5,1.51) (0,.5,1.49) -- (0,.5,1.51) (0,0,1) -- (0,.01,1);
}
\quad&\text{giving a bonus to the complement configuration.}
\end{align*}
Everywhere but on the bottom-most layer, the two penalties precisely cancel; however, on the last layer, only the projection onto $\ket 0$ survives, which thus acts as output penalty once the computation is terminated.

\subsubsection{Multiple Heads Penalty}
Since we only want to allow precisely one head on the computational layer, we will penalize any configuration where two heads are next to each other. This finishes our construction.

\subsection{Valid History State Branching}\label{sec:branching}
In this section, we want to analyse all transition rules and show that the parts where they are ambiguous do not break the evolution of the computation.
First note that all constraints in \cref{sec:static-constraints} are static, i.e.\ there are no possibilities for any ambiguities in the configuration. We will call configurations that obey all those static constraints and have precisely one head symbol on the computational layer---i.e. exactly one of \Sr, \Sl or \textcolor{red}!---\emph{valid} configurations.

We will go through each dynamic penalty in \cref{sec:dynamic-constraints} separately.

\begin{enumerate}
\item
In \cref{sec:moving-constraints}, the transition rules for the faces are unambiguous, since they depend on the red symbol to be in a configuration \inline\Ro.
\item
The rules for moving around a corner, however, \emph{can} happen on a face: in this case, the arrow symbol is moved one layer into the bulk. Observe though that none of the movement transitions can apply to the arrow when it is inside of the bulk (apart from moving it back out with a reverse transition), so the computation branches, but the leg does not proceed: we obtain an evolution of the form
\[
\begin{tikzpicture}
\draw (-.5, 0) -- (11.5, 0);
\foreach \x in {0, 2, ..., 10} {
    \draw (\x, 0) -- (\x+1, 1);
    \pgfmathsetmacro\temp{\x/2+1};
    \draw[fill=white] (\x, 0) circle (2pt) node[xshift=8, yshift=-9] {$\ket{\psi_{\pgfmathprintnumber{\temp}}}$};
    \draw[fill=white] (\x+1, 1) circle (2pt) node[xshift=8, yshift=-9] {$\ket{\psi'_{\pgfmathprintnumber{\temp}}}$};
}
\end{tikzpicture}
\]
where all the primed states are redundant, but at most enlarge the overall evolution by a factor of 2.
\item
In \cref{sec:comp}, the computation transitions are unambiguous; observe in particular that there is no transition rule that simply copies the arrow around the computation edge (by construction, see \cref{sec:moving-constraints} ).
\item
Finally, the input and output constraints are static again.
\end{enumerate}
This allows us to formulate the following two branching lemmas.

\begin{lemma}\label{lem:branching-1}
	Any valid history state for the given transition rules is of size $\BigO(\poly({W,D,H}))$, where $W$, $D$ and $H$ are the cuboid's width, height and depth.
\end{lemma}
\begin{proof}
	Follows by construction; the head can perform at most $\BigO(H\times(W+D))$ unique transitions.
\end{proof}
\begin{lemma}\label{lem:branching-2}
	In case there is more than one head symbol (i.e.\ \textcolor{red}{\textup!}, \Sr or \Sl) present, the minimal valid evolution splits up into $\poly$-sized slices, each of which carries at least one penalty from two directly adjacent heads.
\end{lemma}
\begin{proof}
	The argument is the same as in \cite{Bausch2016}. One can keep all but one of the head symbols fixed; the one left free to move is necessarily meeting another head symbol within $\poly$ many steps.
\end{proof}

\section{QMA-Hardness Proof of Main Theorem}\label{sec:proof}
In this section, we provide a rigorous proof of \cref{th:main}.
Using statical constraints and dynamic rules as in \cref{eq:transition-term,eq:history-state-ham}, we translate the transition rules defined in \cref{sec:cube} into a Hamiltonian $\op H_\text{prop}$, which is geometrically 4-local by construction.

We want to point out that the Hilbert space structure of this lattice Hamiltonian $\op H_\prop$ is not a product space between clock and computation space $\mathcal H_\text{clock}\otimes\mathcal H_\text{comp}$, which would result in a ground state of the standard history state form $\sum_{t}\ket t\ket{\psi_t}$.
The reason for this is that depending on which sub-lattice a spin sits on, its local Hilbert space $\mathcal H_\text{loc}=\field C^4$ decomposes differently.
The red and green spins can be regarded as being completely in the clock space, as all transition rules which act on them are completely classical, i.e.\ they never move any of the red and green spins out of a computational basis state.
The black spins, however, decomposes into a direct sum $\mathcal H_\text{clock}\oplus\field C^2$, the latter space carrying a qubit, and the clock part being reserved for the two arrow symbols \Sl and \Sr, which are part of the clock.

In order to analyse the spectrum, we note that there exists an isometric transformation between our Hamiltonian and Hilbert space, and one which respects the product space structure, which in particular will allow us to regard the Hamiltonian as a ULG Laplacian and apply \cref{lem:kitaev-ulgs}.
Let us be precise at this point,
and use the recently-developed Quantum Thue System terminology defined in \cite[sec.\ 6]{Bausch2016}.
With this new machinery, we can state the following lemma.
\begin{lemma}
The transition rules in \cref{sec:cube} define a Quantum Thue System, and the induced ULG is simple.
\end{lemma}
\begin{proof}
Verifying that the rules define a Quantum Thue System is straightforward by a simple re-ordering of the spins. Simplicity of the corresponding unitary labelled graph follows from \cref{lem:branching-1}; we refer the interested reader to \cite[def.\ 51, lem.\ 52 and 53]{Bausch2016}.
\end{proof}

\noindent
Without further ado, we now proceed to the proof of \cref{th:main}, which we re-state here in a rigorous, but concise fashion.
\begin{theorem}
$\kdHam[4][4]$ is \QMAEXP-complete.
\end{theorem}
\begin{proof}
Containment in \QMAEXP is straightforward, cf.~\cite{Bausch2016}. To show that the Hamiltonian instances of the cube construction define a \QMAEXP-hard family, we will employ techniques proven there which should simplify the analysis.

Let $\L=(\Lyes,\Lno)$ be a \QMAEXP promise problem, as in \cref{def:qmaexp}. By \cref{lem:which-cube}, we know that we can pick a constant error threshold $\delta>0$ such that for any instance $l\in\L$ there exists a cube which allows a verifier circuit for this instance to be executed on the sides. Since we will require probability amplification later on in the proof (\cref{rem:probability-amp}), we set $\delta=f(|l|)$ for some function $f$ to be specified later, and also assume that the original verifier's acceptance probability is $\epsilon_l\le f(|l|)$.

We translate all static and dynamic penalties into a Hamiltonian as explained in \cref{eq:history-state-ham}, and denote the corresponding Hamiltonian operator with
\[
    \op H=\op P+\op H_\prop=\op P_\text{in} + \op P_\text{out} + \op P_\text{static} + \op P_\text{heads} + \op H_\prop,
\]
where $\op P_\text{static}$ comprises all static constraints for the cube (cube structure, binary counter and winding of program), $\op P_\text{heads}$ penalizes any two head symbols next to each other, and such that $\op P_\text{in/out}$ represent the input and output penalties, respectively.

\paragraph{Soundness.}
We first regard the case when $l\in\Lyes$. Denote with $\ket{\Psi_l}$ the valid history state, i.e. the unique uniform superposition ground state of $\op H_\prop$ started out in a valid initial configuration with a single left-moving head in the top left row, and such that no initial or static penalty is violated. Then
\begin{align*}
	\bra{\Psi_l}\op H\ket{\Psi_l} = & \bra{\Psi_l}\op P_\text{in}\ket{\Psi_l}      & =0 & \quad\text{(because $\ket{\Psi_l}$ satisfies all input constraints)} \\
	                                +& \bra{\Psi_l}\op P_\text{out}\ket{\Psi_l}     &    &  \\
	                                +& \bra{\Psi_l}\op P_\text{static}\ket{\Psi_l}  & =0 & \quad\text{($\ket{\Psi_l}$ is valid history state)}                  \\
	                                +& \bra{\Psi_l}\op P_\text{heads}\ket{\Psi_l}   & =0 & \quad\text{($\ket{\Psi_l}$ has one active head symbol)}                \\
	                                +& \bra{\Psi_l}\op H_\prop\ket{\Psi_t}           & =0 & \quad\text{(since $\ket{\Psi_l}$ ground state of $\op H_\prop$)}.
\end{align*}
What remains to be analysed is the output penalty $\bra{\Psi_l}\op P_\text{out}\ket{\Psi_l}$. If we write
$\ket{\Psi_l} = \frac1{\sqrt T}\sum_{t\in T}\ket t\ket{\psi_t}$
where $T$ is the normalization constant for the history state (i.e.\ the number of unique vertices in the ULG evolution represented by $\ket{\Psi_l}$), which we know by \cref{lem:branching-1} to be $T=\BigO(\poly(W, D, H))$---i.e. the number of computational steps taken, including branching, cannot be larger than a polynomial in the cube width, depth and height.
Then 
\begin{align*}
\bra{\Psi_l}\op P_\text{out}\ket{\Psi_l} 
&= \frac1T\left(
    \sum_{t,t'\in T}\bra t\bra{\psi_t}\ketbra T\otimes\Pi_\text{out}\ket{t'}\ket{\psi_{t'}}
\right)
\\
&=
\frac1T \bra{\psi_T}\Pi_\text{out}\ket{\psi_T}
=
\frac1T\mathds P(\text{circuit rejects})
\le \frac1T(\epsilon_l + \delta)=\frac{2f(|l|)}T.
\end{align*}

\paragraph{Completeness.} If $l\not\in\Lyes$, we have to show that for any $\ket\psi$, $\bra\psi\op H\ket\psi$ is bounded away from the \yes-case by a $1/\poly$ gap.
If any static constraint is violated, we can immediately bound $\op H\ge1$. So we can assume that the state $\ket\psi$ is in a valid configuration.

Note that the number of head symbols \Sr, \Sl or \textcolor{red}! is always preserved for any transition rule.
This means that $\op H_\prop$---and therefore also $\op H$---is block-diagonal in the static cube configuration and the number of head symbols on the computational layer, we can regard each case separately.

\begin{enumerate}
\item
In case of multiple head symbols we observe that each head necessarily sweeps the entire surface of the cube.
Mark an arbitrary head symbol, and define $\op H_\prop'$ to be $\op H_\prop$ with any transitions for the other heads removed.
Any such transition rule as in \cref{eq:transition-term} is positive semi-definite, which necessarily means $\op H_\prop\ge\op H_\prop'$ (spectrum wise, by which we mean $\op H_\prop-\op H_\prop'$ is psd itself).
This new operator $\op H_\prop'$ might be non-local, but we only need it to lower-bound the spectrum of $\op H_\prop$.

The marked head symbol will then encounter another head in at most $\poly(W, H, D)$ many steps (it cannot take longer than visiting the entire surface of the cuboid, cf.\ \cref{lem:branching-2}).
At that point, it will pick up a penalty. Utilizing our variant of Kitaev's lemma (\cref{lem:kitaev-ulgs}), we conclude
\[
    \op H=\op P_\text{in/out}+\op P_\text{static}+\op P_\text{heads} + \op H_\prop
    \ge \op P_\text{heads} + \op H_\prop'
    \ge \Omega(1/\poly(W, H, D)).
\]
We thus need to set $f$ to a function which allows a polynomial separation (in the system size) between \yes and \no instance; by \cref{rem:cube-accuracy}, this is always possible.
\item
The same argument lets us bound $\op H\ge \op P_\text{in/out}+\op H_\prop=\Omega(1/\poly(W, H, D))$ in case of a single head valid history state since $l$ is a \no-instance.
\item
What remains to be analysed is the zero head case. There are two standard approaches: we can either increase the number of symbols on the computational layer, such that on one side of a head---i.e. behind it in direction of the computation---we take one kind of symbols, and on the other side we take the other set; constructions like this can be constrained by a regular expression (without repetition of symbols, cf. \cite[lem. 5.2]{Gottesman2009}) and thus penalized with local terms.

Since our benchmark tries to reduce the local dimension of the system, we instead add a bonus term $\op B$ to the Hamiltonian $\op H$, and such that $\op B\ket\psi=-g(|l|)\times h\ket\psi$ where $h$ is the number of head symbols for any basis state $\ket\psi$ of $\op H$, and $g$ is a function chosen such that there is again a $1/\poly$ separation between the zero head state and the ground state for \yes-instances, but such that multiple head configurations stay bounded away from $0$.
It is clear that $\op B$ can be implemented by a 1-local term of the form $-g(|l|)\ketbra{\text{head}}$.

To be more precise and to determine how quickly $g$ has to grow, assume that the construction up to now satisfies $\lmin\le1/A$ for the \yes case, and $\lmin\ge1/B$ in the \no case (excluding zero heads), where $B=\BigO(\poly(W, H, D))$, and $A\ge 4B\times W\times H\times D$. Choose $g=2/A$; since there can be at most $W\times H\times D$ heads on the cuboid's faces, we obtain the bounds
\begin{align*}
    \lmin\begin{cases}
        \le 1/A - 2/A = -1/A & \text{if } l\in\Lyes \\
        \ge 1/B - 4(W\times H\times D)/A \ge 1/B - 1/B \ge 0 & \text{otherwise, and for at least 1 head.}
    \end{cases}
\end{align*}
The zero head case can then easily be lower-bounded by $\op H\ge0$.
\end{enumerate}
We have thus shown a promise gap of $1/\poly$ in the system size: for $l\in\Pi_\yes$, $\op H\le-\Omega(1/\poly(W, H, D))$, and $\op H\ge0$ otherwise.
The claim of \cref{th:main} follows.   
\end{proof}

\end{document}